\documentclass[11pt]{article}
\usepackage{amsmath, amsthm, amssymb}
\usepackage{graphicx}
\usepackage[latin1]{inputenc}
\usepackage{fullpage}
\usepackage[pdftitle={Tile},
            pdfauthor={Author},	
            pdfkeywords={Keywords}]{hyperref}
\usepackage{color}
\usepackage{bm}
\usepackage{algorithmicx}
\usepackage[noend]{algpseudocode}
\usepackage{algorithm}

\usepackage{xspace}

\newtheorem{corollary}{Corollary}

\newtheorem{lemma}{Lemma}
\newtheorem{remark}{Remark}
\newtheorem{theorem}{Theorem}

\newtheorem{prop}{Proposition}
\newtheorem{fact}{Fact}

\newcommand{\R}{\mathbb{R}}

\newcommand{\OPT}{\textsc{OPT}}

\newcommand{\N}{\mathbb{N}}

\newcommand{\remove}[1]{}
\newcommand{\nullvec}{\mathbf{0}}
\newcommand{\calC}[0]{{\cal C}}

\newcommand{\onevec}{\mathbf{1}}
\newcommand{\evec}{\mathbf{e}}
\newcommand{\dvec}{\mathbf{d}}
\newcommand{\uvec}{\mathbf{u}}
\newcommand{\vvec}{\mathbf{v}}
\newcommand{\wvec}{\mathbf{w}}
\newcommand{\cvec}{\mathbf{c}}
\newcommand{\todo}[1]{\marginpar{\tiny \flushleft{#1}}}

\title{Approximation Algorithms for Clustering via Weighted  Impurity Measures}
\author{Ferdinando Cicalese\\  University of Verona, Italy\\
{\tt ferdinano.cicalese@univr.it} \and 
Eduardo Laber \\PUC-Rio, Brazil\\ {\tt laber@inf.puc-rio.br}}
\date{}

\begin{document}
	\maketitle

\begin{abstract}
An impurity measures  $I: \mathbb{R}^k \mapsto  \mathbb{R}^+$ is a function that assigns a $k$-dimensional vector $\vvec$ to a non-negative value $I(\vvec)$ so that the more homogeneous $\vvec$, with respect to the values of its coordinates, the larger its impurity. Well known examples of impurity measures are the Entropy and the Gini impurities. We study the problem of clustering based on impurity measures:
given a collection of $n$ many $k$-dimensional vectors $V \subset \N^k$ with non-negative integer coordinates and an impurity measure $I$, the goal is to find a  partition ${\cal P}$ of $V$ into $L$  groups  $V_1,\ldots,V_L$   so as to minimize the  sum of the  impurities of the groups in ${\cal P}$, i.e.,  $I({\cal P})= \sum_{m=1}^{L} I\bigg(\sum_{ \vvec \in V_m} \vvec \bigg).$

Impurity minimization has been widely  used as quality assessment measure in probability distribution clustering as well as in categorical clustering where it is not possible to rely on geometric properties of the data set. However, in contrast to the case of metric based clustering, the current  knowledge of impurity measure based clustering in terms of  approximation and inapproximability results is very limited. 

Our research contributes to fill this gap.  We first present a simple linear time algorithm that  simultaneously achieves $3$-approximation for the Gini impurity measure and $O(\log \sum_{\vvec \in V} \|\vvec\|_1)$-approximation for the Entropy impurity measure. Then, for the Entropy impurity measure---where we also show that finding the optimal clustering is strongly NP-hard---we are able to design a polynomial time  
$O\big(\log^2(\min\{k,L\})\big)$-approximation algorithm.  Our algorithm relies on a nontrivial characterization  of a class of clusterings that necessarily includes a partition achieving $O\big(\log^2(\min\{k,L\})\big)$--approximation of the impurity of the optimal partition. Remarkably, this is the first polynomial time algorithm with approximation guarantee independent of the number of points/vector and not relying on any restriction on the components of the vectors
for producing clusterings with minimum entropy.
\end{abstract}

\pagebreak
\section{Introduction}
Data clustering is a fundamental tool in the analysis of large datasets 
to reduce the computational resources required to handle them. 
For a recent comprehensive description of different clustering methods and their applications we refer to \cite{HMMR15}.
In general, clustering is the problem of partitioning a set of multidimensional points so that, in the output partition, 
similar points are grouped together and 
dissimilar points are separated. 
When data are numerical, the quality of a clustering is naturally based on the pairwise distance 
in some metric space where the 
data points lie. In many applications, however, data is categorical, i.e, data points are described by binary attributes or qualitative features, e.g., 
ethnicity, music preferences, place of residence, hair color, etc.  
In categorical clustering, rather than employing distance functions to measure the pairwise distance of data points, 
many clustering algorithms rely on 
so called impurity measures, that estimate the "purity" of a set of data points. 
A partition is then evaluated by considering the total impurity of 
the sets in which it splits the data set.

The design of clustering methods based on   impurity measures is the central theme of this paper. 
More formally, an impurity measures  $I: \vvec \in \mathbb{R}^k \mapsto I(\vvec) \in \mathbb{R}^+$ is 
a function that assigns a vector $\vvec$ to a non-negative value $I(\vvec)$ so that 
the more homogeneous $\vvec$, with respect to the values of its coordinates, the larger its impurity.
Well known examples of impurity measures\footnote{In 
the literature, also the form without the scaling factor $\|\vvec\|_1$ is found and the function used here
is also referred to as scaled/weighted impurity. In the section Preliminaries we will give a more 
precise definition of a wider class of impurity measures.} 
are the Entropy impurity \cite{conf/cikm/BarbaraLC02,conf/icml/LiMO04} and the Gini impurity \cite{Breiman84}:

$$ I_{Ent}(\vvec) = \|\vvec \|_1  \sum_{i=1}^k \frac{v_i}{\| \vvec\|_1} \log \frac{\| \vvec\|_1}{ v_i} \quad \mbox{ and } \quad 
 I_{Gini}(\vvec) =  \|\vvec \|_1 \sum_{i=1}^k \frac{v_i}{\| \vvec\|_1} \left ( 1- \frac{v_i}{\| \vvec\|_1}  \right ) $$

\noindent
{\bf Problem Description.}
We are given a collection of $n$ many $k$-dimensional vectors $V \subset \N^k$ with non-negative
integer coordinates and we are also given an impurity measure $I$. The goal is to
find a  partition ${\cal P}$ of $V$ into $L$ disjoint groups of 	 vectors $V_1,\ldots,V_L$   so as to minimize
the  sum of the  impurities of the groups in ${\cal P}$, i.e.,  $$I({\cal P})= \sum_{m=1}^{L} I\bigg(\sum_{ \vvec \in V_m} \vvec \bigg).$$ 
We refer to this problem as  the {\sc Partition with Minimum
Weighted Impurity Problem} ({\sc PMWIP}). 

The complexity of clustering in metric spaces, e.g., in the case of k-median, k-center and k-means,
is well understood from the perspective of approximation algorithms in the sense that  the gap between the ratios
achieved by the best known  algorithms and the largest known inapproximability 
factors, assuming $P\ne NP$, are somehow tight (see \cite{journals/corr/AwasthiCKS15} and references therein).
In contrast, despite its wide use in applications, the understanding of clusterings based on impurity measures  is much more limited  as we detail further.
As an example, for information theoretic clustering \cite{dmk-difcatc-03}, which is closely
related to clustering based on entropy impurity, no hardness result
beyond the NP-Completeness proved in \cite{journals/eccc/AckermannBS11} is available
 and the best  approximation 
ratio known, when no assumptions on the 
the data input are made and all probability distributions have the same weight,  is $O(\log n)$, where $n$ is the number of points to be clustered \cite{cm-fmsitc-08}. 

In this paper we present algorithms and complexity results that contribute to the understanding of  clustering
based on impurity measures. In particular, for our formulation of the information theoretic clustering  problem we manage to obtain an approximation factor that depends on the logarithmic of the number of clusters rather than on the number $n$ of data points.


\medskip

\noindent
{\bf Our Results and Techniques.}
First we present a simple linear time algorithm
that simultaneously guarantees a $3$-approximation for the $I_{Gini}$
and an $O(\log \sum_{\vvec \in V} \|\vvec\|_1 )$ approximation for $I_{Ent}$.
In addition, for the relevant case where all vectors in $V$ have the same $\ell_1$ norm the
algorithm provides an $O( \log n + \log k)$ approximation.
Then, we present a second algorithm that provides 
 an $O( \log^2 (\min \{k,L\} ))$-approximation for 
$I_{Ent}$ in polytime. We also prove that the problem considered here is strongly NP-Hard for $I_{Ent}$.

When $k>L$, both algorithms employ an extension of the technique  introduced in \cite{conf/icml/LMP18} that allows  to reduce the dimensionality of the
vectors in $V$ from $k$ to $L$  with a controllable additive loss in the approximation ratio. In \cite{conf/icml/LMP18}, where the case $L=2$ is studied, after the reduction step, an optimal clustering algorithm is used. However,  for arbitrary  $L$, the same strategy cannot be applied since
the problem is NP-Complete. Thus, it is crucial to devise novel procedures to handle the case where $k \le L$.

The procedure employed by the first algorithm is quite simple: it assigns vectors to groups according to the 
dominant coordinate, that is,  one with the largest value. The procedure of the second algorithm is
more involved, it relies on
the combination of the following results: (i) the existence of an optimal algorithm for $k=2$ \cite{journals/tit/KurkoskiY14};(ii) the existence of a mapping $\chi:\R^k \mapsto \R^2$ 
such that 
for a set of vectors $B$ which is pure, i.e., a set of vectors with the same dominant component,  
$I_{Ent}(\sum_{\vvec \in B} \vvec) = O(\log k) I_{Ent}(\sum_{\vvec \in B} \chi(\vvec))$  and (iii) 
a structural theorem that states that there exists a partition whose impurity is at an $O(\log ^2 k)$ factor from the optimal one 
and such that at most one of its groups is mixed, i.e., it is not pure, namely contains vectors with different dominant coordinate. 
The items (i) and (ii) would be sufficient to obtain a $\log k$ approximation had
we had $ I_{Ent}( \sum_{\vvec \in S} \vvec ) = O(\log k) \cdot I_{Ent}( \sum_{\vvec \in S} \chi(\vvec)) $ for all sets $S \subseteq V$.
However, this property does not hold for arbitrary $S$ but it does for sets $S$ that are
not mixed. This is the reason why item (iii) is important - it allows to only consider 
partitions in which at most one group is mixed. The search for a partition 
of this type with low impurity can be  achieved in pseudo-polynomial time via
Dynamic Programming. To obtain a polynomial time algorithm we then employ a filtering technique similar
to that employed for obtaining a FPTAS for the subset sum problem.

\medskip

\noindent
{\bf Related Work.}
Partition optimization based on impurity measures as defined in the {\sc PMWIP} is also employed in the construction of 
decision trees/random forests to asses the quality of nominal attributes in the attribute selection step (see, e.g., 
\cite{Breiman84,BPKN:92,Chou91,journals/datamine/CoppersmithHH99,conf/icml/LMP18,journals/tit/KurkoskiY14}, and 
references quoted therein). Kurkoski and Yagi  \cite{journals/tit/KurkoskiY14} showed that for the entropy impurity measure the problem
can be solved in polynomial time   when $k=2$. The correctness of this algorithm relies
on a theorem,  proved in \cite{Breiman84},  which is  generalized for $k > 2$  and 
 $L$ groups in \cite{Chou91,BPKN:92,journals/datamine/CoppersmithHH99}.
Basically, these theorems  state that there exists an optimal solution that can be separated by hyperplanes
in $\R^k$.  These results imply the existence of $O(n^k)$ optimal algorithm when $L=2$. 
Recently, one of authors proved that the problem is $NP$-Complete for $I_{Ent}$, even when $L=2$, and
presented constant approximation algorithms for a class of impurity measures that 
include Entropy and Gini for $L=2$ \cite{conf/icml/LMP18}.

Kurkoski and Yagi  \cite{journals/tit/KurkoskiY14} also observe that the ({\sc PMWIP}) with the Entropy impurity measure 
corresponds to the problem of designing a quantizer for the output $Y$
of a discrete memoryless channel  in order to maximize the 
mutual information between the channel's input $X$ and the quantizer's output $Z$. This 
problem (also motivated by the construction of polar codes) has recently attracted large interest in the information theory community
 \cite{journals/tit/TalV13,journals/tit/KurkoskiY14,journals/corr/KartowskyT17a,journals/tit/PeregT17,conf/NazerOP17}. 
The correspondence to {\sc PMWIP} is obtained by taking the channel inputs $X$ as the components of the vectors in $V$, 
the channel's output $Y$ as the points/vectors in 
$V$ and the quantizer's outputs $Z$ as the clusters. The impurity of the clustering coincides with the conditional entropy
$H(X \mid Z).$
The focus in the Information Theory community is proving bounds, as a function of $|X|$, $|Y|$ and $|Z|$,
 on the mutual information degradation due to  quantization/clustering rather than designing approximation algorithms.

\remove{
 In \cite{journals/corr/KartowskyT17a}, a simple greedy procedure is analyzed in terms of additive approximation 
of the degradation of  mutual information $I(X;Y) - I(X;Z).$ 
Therefore, while in terms of optimal value the goal of ({\sc PMWIP}) and maximum mutual information quantization coincide, 
the approximation 
criterion of the above papers is different and does not allow for a direct comparison. 
}

Another problem which is strictly related to {\sc PMWIP} with the entropy impurity  measure is the 
problem of clustering probability distributions  based on the 
 Kullback-Leibler (KL) divergence.  
In particular, the $MTC_{KL}$ problem of \cite{cm-fmsitc-08} asking for the clustering of a set of $n$ probability distributions of dimension $k$
into $L$ clusters minimizing the total Kullback-Leibler (KL) divergence of the points from the centroids of the clusters, 
corresponds to the particular case of {\sc PMWIP} where each vector has the same $\ell_1$ norm.
In \cite{cm-fmsitc-08} an $O(\log n)$ approximation for $MTC_{KL}$ is given. 
Under the additional assumption that  every element of every probability distribution is larger than a constant,
Ackermann et. al. \cite{Ackermann:2008:CMN:1347082.1347170,conf/soda/AckermannB09,journals/talg/AckermannBS10} presents an $(1+\epsilon)$-approximation algorithm for $MTC_{KL}$
that runs in $O(nkL + k 2^{O(L / \epsilon)} \log^{L+2} (n))$ time.  By using similar assumptions on the components of the input probability distributions, 
Jegelka et. al.  \cite{journals/corr/abs-0812-0389} show that Lloyds 
$K$-means algorithm---which is however also exponential time in the worst case 
\cite{Vattani2011}---obtains an $O( \log L)$ approximation for
$MTC_{KL}$. 

It has to be noted that although the optimal solutions of  $MTC_{KL}$  and {\sc PMWIP} are the same, 
the problems differ with regard to the  approximation guarantee pursued. Let  $\OPT(V)$ denote the minimum possible impurity 
of a partition of the input set of vectors $V$ and $I({\cal P}(V))$ the impurity of partition ${\cal P}$ of $V$.  
The approximation goal in our study of {\sc PMWIP} is to find a partition ${\cal P}$ that minimizes the ratio  $I({\cal P}(V)) / \OPT(V)$
while in the case of the above papers on $MTC_{KL}$ the goal is to find ${\cal P}$ that  minimizes 
$(I({\cal P})- H) / (\OPT(V)-H)$, where $H$ is
the sum of the entropies of the vectors in $V$. Among the algorithms mentioned for
$MTC_{KL}$, the one that allows a more direct  comparison with ours is the one proposed
in \cite{cm-fmsitc-08} since it runs in polytime and does not rely on assumptions
over the input data. An $\alpha$-approximation for the $MCT_{KL}$ problem implies $\alpha$-approximation 
for the special case of $PMWIP$ with vectors of the same $\ell_1$ norm, so the approximation measure used in \cite{cm-fmsitc-08} 
is more general. 
However, our results apply to a more general problem and nonetheless we are able to provide approximation guarantee 
depending on the logarithm of the number of clusters while the guarantee in \cite{cm-fmsitc-08} 
depends on the logarithm of the number of input vectors.

\remove{
The algorithm from \cite{}
has the adva
Thus, an $\alpha$-approximation for
$MTC_{KL}$ implies a $\alpha$-approximation for the special case of  {\sc PMWIP} with uniform vectors but the converse
is not necessarily true. 
}

\section{Preliminaries}
We start defining some notations employed throughout the paper.
An instance of \textsc{PMWIP} is a  triple  $(V,L,I)$, where 
 $V$ is a collection of non-null vectors in $\R^k$ with non-negative
integer coordinates, $L$ is an integer larger than $1$  and $I$ is a scaled impurity measure.

We assume that for each coordinate  $i=1, \dots, k$ there exists at 
least one vector $\vvec \in V$ whose $i$th coordinate is non-zero, i.e., 
the vector $\sum_{\vvec \in V} \vvec$ has no zero coordinates---for otherwise
we could consider  an instance of \textsc{PMWIP} with the vectors lying in  some dimension $k' < k$. 
For a set of vectors $S$,
the impurity $I(S)$ of $S$ is given by
$I( \sum_{\vvec \in S} \vvec)$. The impurity of a partition ${\cal P}=(V^{(1)}, \ldots,V^{(L)})$ of the set $V$ is then
$I({\cal P}) =\sum_{i=1}^L I(V^{(i)})$.  We use $\OPT(V, I, L)$ to denote the minimum possible impurity for
an $L$-partition of $V$ and, whenever the context is clear, we simply talk about instance $V$ 
(instead of $(V,I,L)$) and of the impurity of an optimal solution as $\OPT(V)$ (instead of $\OPT(V,I,L)$).
We say that a partition $(V^{(1)},\ldots,V^{(L)})$ 
is optimal for input $(V,L,I)$ iff $\sum_{i=1}^L I(V^{(i)})=\OPT(V,I,L)$.

For an algorithm ${\cal A}$ and an instance $(V, I,L)$, we denote by ${\cal A}(V, I,L)$ and $I({\cal A}(V, I,L))$ the 
partition output by ${\cal A}$ on instance $(V,I,L)$ and its impurity, respectively. 
Whenever it is clear from the context, we omit to specify the instance and write $I({\cal A})$ for 
$I({\cal A}(V,I,L)).$


We use bold face font to denote vectors, e.g., $\uvec, \vvec, \ldots$. 
For a vector $\uvec$ we use $u_i$ to denote its $i$th component.
Given two vectors $\uvec=(u_1,\ldots,u_k)$ and
$\vvec=(v_1,\ldots,v_k)$ we use $\uvec \cdot \vvec$ to denote their inner product and 
$\uvec \circ \vvec=(u_1 v_1,\ldots,u_k v_k)$ to denote their
component-wise (Hadamard) product. 
We use  $\nullvec$ and $\onevec$ to denote  the vectors in $\R^k$ with all coordinates equal to 0 and 1,
respectively.
 We use $[m]$ to denote the set of the first $m$ positive integers.
For $i=1, \dots, k$ we denote by $\evec_i$ the vector in $\R^k$ with the $i$th coordinate equal to $1$ and 
all other coordinates equal to $0$.

The following properties will be useful in our analysis.

\begin{prop}
Let $p \in [0,0.5)$. Then, $p\log(1/ p) \ge (1-p)\log[1/(1-p)]$ 
\label{prop:basic-inequality}
\end{prop}
\begin{proof}
It is enough to show that $g(p)=p\log(1/ p) - (1-p)\log[1/(1-p)] \geq 0$ in the interval $[0,1/2].$
For this, simply observe that $g(0)=g(1/2) = 0$ and that $g(p)$ is concave in the interval 
$[0,1/2]$ (the second derivative is negative).
%
\end{proof}


\begin{prop}
Let $A >0$. The function $f(x)= x \log (A /x)$ is increasing in the interval 
$(0,A/e]$ and decreasing in the interval $(A/e,A]$ so that
its maximum value in the interval $[0,A]$ is  $(A \log e)/e.$
\label{prop:increasing-new}
\end{prop}
\begin{proof}
The result follows because   $f'(x)= (\ln(A) - \ln x - 1)/ \ln 2 $, the derivative of $f(x)$ is positive
in the interval $(0,A/e)$ and negative in the interval  $[A/e,A)$.
\end{proof}

\subsection{Frequency weighted impurity measures with subsystem property}
\label{sec:impmeasures}
The impurity measures we will focus on, namely Gini and Entropy, are special cases  of a larger class of impurity measures, which we denote by ${\cal C}$, that satisfy the following definition
	\begin{align}
		I(\uvec) = \|\uvec\|_1 \sum_{i=1}^{dim(\uvec)} f\bigg(\frac{u_i}{\|\uvec \|_1} \bigg),
	\tag{P0}
	\end{align}
	where $dim(\uvec)$ is the dimension of vector $\uvec$ and $f : \R \rightarrow \R$ is a function satisfying the following conditions:

\begin{enumerate}
\item $f(0)=f(1)=0$ \hfill (P1)\vspace{-5pt}
\item $f$  is strictly concave in the interval [0,1] \hfill(P2)\vspace{-5pt}
\item For all $0 < p \le q \le 1,$ it holds that  
$f(p) \le \frac{p}{q} \cdot f(q) + q\cdot  f\left(\frac{p}{q}\right)$ \hfill{(P3)}
\end{enumerate}

Impurity measures satisfying the  conditions (P0)-(P2) are called \emph{frequency-weighted 
impurity measures based on concave functions}~\cite{journals/datamine/CoppersmithHH99}.
A fundamental properties of such impurities measures is that they are {\em superadditive} as shown 
in \cite{journals/datamine/CoppersmithHH99}. We record this property in the following lemma.


\begin{lemma}[Lemma 1 in \cite{journals/datamine/CoppersmithHH99}]
If $I$ satisfies (P0)-(P2) then for
every  vectors  $\uvec_L$ and $\uvec_R$ in $\R^k_+$,  we have
 $I(\uvec_L+ \uvec_R) \ge I(\uvec_L)+ I(\uvec_R)$.
\label{lem:imputity-supperadditivy}
\end{lemma}


The Entropy and the Gini impurity measure satisfy the definition (P0) by means of the functions
$f_{Entr}(x) = -x \log x$ and $f_{Gini}(x) = x(1-x)$. In fact, for a vector $\uvec \in \R^k$ the Entropy   
impurity $I_{Ent}(\uvec)$ and the Gini impurity $I_{Gini}(\uvec)$ are defined by 
\begin{equation} \label{entropy-gini}
I_{Ent}(\uvec) = \|\uvec\|_1 \sum_{i=1}^k f_{Entr}\left(\frac{u_i}{\|\uvec\|_1}\right)
\quad \mbox{ and } \quad
I_{Gini}(\uvec) = \|\uvec\|_1 \sum_{i=1}^k f_{Gini}\left(\frac{u_i}{\|\uvec\|_1}\right).
\end{equation}
It is also easy to see that 
$I_{Ent}(\uvec) = \|\uvec\|_1 H\left(\frac{u_1}{\|\uvec\|_1}, \frac{u_2}{\|\uvec\|_1}, \dots, \frac{u_k}{\|\uvec\|_1}\right)$ 
where $H(\cdot)$ denotes the Shannon entropy function. 

The following fact states that both the Gini and Entropy impurity measures belong to the class $\calC$. For the sake of 
self-containment we have deferred a simple proof of this fact to the appendix.

\begin{fact} \label{fact:gini-entropy}
Both $f_{Entr}$ and $f_{Gini}$ satisfy properties (P1)-(P3), and, in particular, we have that $f_{Entr}$ satisfies (P3)  with equality. 
Therefore both the Gini impurity measure $I_{Gini}$ and 
 the Entropy impurity measure $I_{Ent}$ 
belong to $\calC$.

\end{fact}
%
%

We now show that  the impurity measures of class ${\cal C}$
satisfy a special {\em subsystem property} which will be used in our analysis to 
relate the impurity of partitions for instances of dimension $k$
with the  impurity of partitions for instances of dimension $L$.

\begin{lemma}[Subsystem Property]
Let $I$ be an impurity measure in $\calC$.
Then, for every $\uvec \in \R^k_+$ and pairwise orthogonal vectors $\dvec^{(1)}, \dots, \dvec^{(L)} \in \{0,1\}^k$,
such that $\sum_{i=1}^L \dvec^{(i)} = \onevec,$ we have
 \begin{equation} I(\uvec) \le I \bigg ( (\uvec \cdot \dvec^{(1)}, \uvec \cdot \dvec^{(2)}, \dots,  \uvec \cdot \dvec^{(L)}) \bigg) 
 + \sum_{i=1}^L I( \uvec \circ \dvec^{(i)}).
\label{eq:subsystem}
\end{equation}
Moreover, for $I = I_{Ent}$ we have that (\ref{eq:subsystem}) holds with equality.
 \label{lemma:subsystem}
\end{lemma}
\begin{proof}
Let $f$ be the concave function used by the frequency-weighted impurity measure $I$. 

For $i=1, \dots, L,$ let $\uvec^{(i)} = \uvec \circ \dvec^{(i)}.$ We have
\begin{eqnarray}
I(\uvec) &=& \|\uvec\|_1 \sum_{j=1}^k f(\frac{u_j}{\|\uvec\|_1})  \label{subsys-eq-1}\\
&=& \|\uvec\|_1 \sum_{i=1}^L \sum_{j \mid d^{(i)}_j = 1} f(\frac{u_j}{\|\uvec\|_1}) \label{subsys-eq-2} \\
&\leq&  \|\uvec\|_1 \sum_{i=1}^L \sum_{j \mid d^{(i)}_j = 1} 
\frac{u_j}{\|\uvec\|_1} \frac{\|\uvec\|_1}{\|\uvec^{(i)}\|_1} f(\frac{ \| \uvec^{(i)} \|_1}{\|\uvec\|_1}) +
\frac{ \| \uvec^{(i)} \|_1}{\|\uvec\|_1} f(\frac{u_j}{\|\uvec^{(i)}\|_1}) \label{subsys-eq-3}\\
&=& \|\uvec\|_1 \sum_{i=1}^L \sum_{j \mid d^{(i)}_j = 1} 
\frac{u_j}{\|\uvec^{(i)}\|_1} f(\frac{ \| \uvec^{(i)} \|_1}{\|\uvec\|_1}) +
\sum_{i=1}^L \sum_{j \mid d^{(i)}_j = 1} \| \uvec^{(i)} \|_1   f(\frac{u_j}{\|\uvec^{(i)}\|_1}) \label{subsys-eq-4}\\
&=&  \|\uvec\|_1 \sum_{i=1}^L   f(\frac{ \| \uvec^{(i)} \|_1}{\|\uvec\|_1}) +
\sum_{i=1}^L   \| \uvec^{(i)} \|_1  \sum_{j \mid d^{(i)}_j = 1}  f(\frac{u_j}{\|\uvec^{(i)}\|_1}) \label{subsys-eq-5}\\
&=& I\left( (\uvec \cdot \dvec^{(1)}, \uvec \cdot \dvec^{(2)}, \dots, \uvec \cdot \dvec^{(L)})\right) +
\sum_{i=1}^L I(\uvec \circ \dvec^{(i)}) \label{subsys-eq-6}
\end{eqnarray}
where (\ref{subsys-eq-2}) follows from (\ref{subsys-eq-1}) by splitting the second summation according to the partition 
of $[k]$ induced by the 
non zero components of the vectors $\dvec^{(i)}$; 
(\ref{subsys-eq-3}) follows from (\ref{subsys-eq-2}) by applying property (P3) with $p = \frac{u_j}{\|\uvec\|_1}$ and 
$q = \frac{\|\uvec^{(i)}\|_1}{\|\uvec\|_1}$; (\ref{subsys-eq-4}) follows from (\ref{subsys-eq-3}) by simple algebraic manipulations;
(\ref{subsys-eq-5}) follows from (\ref{subsys-eq-4}) since by definition of $\uvec{(i)}$ we have 
$\sum_{j \mid d^{(i)}_j = 1} {u_j}  = \|\uvec^{(i)}\|_1$;  (\ref{subsys-eq-5}) follows from (\ref{subsys-eq-4}) since
$\|\uvec^{(i)}\|_1 = \uvec \cdot \dvec^{(i)}$ and 
$I(\uvec \circ \dvec^{(i)}) = \sum_{j \mid d^{(i)}_j = 1} \|\uvec \circ \dvec^{(i)}\|_1 f(\frac{u_j}{\| \uvec \circ \dvec^{(i)}\|_1})$
and $ \|\uvec \circ \dvec^{(i)}\|_1 = \|\uvec^{(i)}\|_1.$

\medskip
The second statement of the lemma follows immediately by the fact that the concave function $f_{Entr}$ satisfies property (P3) with 
equality (see Fact \ref{fact:gini-entropy}). Hence, for $I_{Ent}$ the inequality in (\ref{subsys-eq-3}) becomes an equality.

\end{proof}

\begin{remark}
The Subsystem property in the previous lemma 
holds also under the stronger assumption that vectors $\dvec$'s are from $[0,1]^k$ and not necessarily orthogonal.
\end{remark}



\section{Handling high dimensional vectors} \label{sec:general}
In this section we present an approach  to
address instances $(V,I,L)$ with $I \in {\cal C}$ and $k > L.$
 It consists of two steps: 
finding a 'good'  projection of $\R^k$ into $\R^L$ and  then
solving \textsc{PMWIP} for the projected  instance with $k= L$.
Thus, in the next sections we will be focusing on how to build this projection
and  how to solve instances with $k \le L$.
The material of this section is a generalization  for arbitrary $L$ of the results introduced in \cite{conf/icml/LMP18} for $L=2$.

Let ${\cal D}$ be the family of all sequences $D$ 
of  $L$ pairwise orthogonal directions in $\{0,1\}^k$, such that $\sum_{\dvec \in D} \dvec = \onevec.$
For each $D=(\dvec^{(1)}, \dots,  \cdot \dvec^{(L)}) \in {\cal D}$  and any $\vvec \in \R^k$ we  define  the operation
 $collapse_{D} : \R^k_+ \rightarrow \R^L_+$ by 
$$collapse_{D} (\vvec) = (\vvec \cdot \dvec^{(1)}, \dots, \vvec \cdot \dvec^{(L)}).$$

We also naturally extend the operation to sets of vectors $S$, by defining $collapse_{D}(S)$ as the multiset of 
vectors obtained by applying  $collapse_{D}$ to each vector of $S$.

Let ${\cal A}$ be an algorithm 
that on instance  $(V, I, L)$ chooses a sequence of vectors $D = \{\dvec^{(1)}, \dots, \dvec^{(L)}\} \in {\cal D}$ 
and returns a partition $(V^{(1)},\ldots,V^{(L)})$ such that
$(collapse_{D}(V^{(1)}),\ldots,collapse_{D}(V^{(L)}))$ is a 'good' partition for the
 $L$-dimensional instance $(collapse_{D}(V), I, L)$. In this section we quantify the relationship between the
 approximation attained by $(collapse_{D}(V^{(1)}),\ldots,collapse_{D}(V^{(L)}))$ for instance $(collapse_{D}(V), I, L)$ and the 
 corresponding approximation attained by $(V^{(1)},\ldots,V^{(L)})$ for instance $(V, I, L)$.

\medskip 

\noindent
 Let $\uvec^{(i)} = \sum_{ \vvec \in V^{(i)}} \vvec$.
From the subsystem property we have the following upper bound on the impurity of the partition returned by ${\cal A}.$

$$ I({\cal A})= \sum_{i=1}^L  I(\uvec^{(i)}) 
\le \sum_{i=1}^L  I \left ( (\uvec^{(i)} \cdot \dvec^{(1)},\ldots,\uvec^{(i)} \cdot \dvec^{(L)}) \right ) 
+   \sum_{i=1}^L \sum_{\dvec \in D} I( \uvec^{(i)} \circ \dvec)  $$
Thus, by the superadditivity  of  $I$ we have
\begin{equation} \label{eq:collapse}
I({\cal A}) \le  \sum_{i=1}^L  I \left ( (\uvec^{(i)} \cdot \dvec^{(1)},\ldots,\uvec^{(i)} \cdot \dvec^{(L)}) \right )
+  \sum_{\dvec \in D} I( \uvec \circ \dvec).
\end{equation}

\medskip

We now show two lower bounds on $\OPT(V,I,L).$ For the sake of simplifying the notation we will use $\OPT(V)$ for $\OPT(V,I,L).$

\begin{lemma} \label{lemma:optCollapse}
For any instance $(V, I, L)$ of \textsc{PMWIP} 
and any $D = \{\dvec^{(1)}, \dots, \dvec^{(L)}\}\in {\cal D}$  we have
	$\OPT(V) \ge \OPT(collapse_{D}(V)).$
\end{lemma}

\begin{proof}
Let $V^{(1)}, \dots, V^{(L)}$ be an optimal partition for $V$, i.e., 
\begin{equation}
\sum_{i=1}^L I(\sum_{\vvec \in V^{(i)}} \vvec) = \OPT(V).
\end{equation}
We define the corresponding partition on the vectors $\tilde{\vvec}$ in $collapse_{D}(V)$ by 
letting $\tilde{V}^{(i)} = \{collapse_D(\vvec)  \mid \vvec \in V^{(i)}\}$. We have 
\begin{equation} \label{5-step1}
\sum_{i = 1}^L I(\sum_{\tilde{\vvec} \in \tilde{V}^{(i)}} \tilde{\vvec}) \geq \OPT(collapse_D(V)).
\end{equation}

Let $\uvec^{(i)}= \sum_{\vvec \in V^{(i)}} \vvec $.
Moreover, by the subadditivity of $f$, we have that for each $i=1, \dots, L,$ it holds that 
\begin{align}
I(\sum_{\vvec \in V^{(i)}} \vvec) = ||\uvec^{(i)}||_1 \sum_{j=1}^k f \left (\frac{\uvec^{(i)}_j}{||\uvec^{(i)}||_1}\right ) = & ||\uvec^{(i)}||_1 \sum_{j=1}^L  \sum_{ \ell | d^{(j)}_{\ell}=1 } 
f \left (\frac{\uvec^{(i)}_{\ell}}{||\uvec^{(i)}||_1} \right ) \ge \notag \\
||\uvec^{(i)}||_1 \sum_{i=1}^L f\left(\frac{\sum_{\ell \mid d^{(j)}_{\ell} = 1} \uvec^{(i)}_{\ell}}{||\uvec^{(i)}||_1} \right)
 = I(\sum_{\tilde{\vvec} \in \tilde{V}^{(i)}} \tilde{\vvec}) & \notag
 \end{align}

\remove{
$$I(\sum_{\vvec \in V^{(i)}} \vvec) = ||\vvec||_1 \sum_{j=1}^k f(\frac{v_j}{||\vvec||_1}) \geq
||\vvec||_1 \sum_{i=1}^L f\left(\frac{\sum_{j \mid d^{(i)}_j = 1} v_j}{||\vvec||_1} \right)  = I(\sum_{\tilde{\vvec} \in \tilde{V}^{(i)}} \tilde{\vvec}),$$
}

which implies $$\OPT(V) = \sum_{i=1}^L I(\sum_{\vvec \in V^{(i)}} \vvec) \geq \sum_{i=1}^L I(\sum_{\tilde{\vvec} \in \tilde{V}^{(i)}} \tilde{\vvec})$$
that combined with (\ref{5-step1})  gives the desired result. 
\end{proof}

The following result, proved in \cite{BPKN:92,journals/datamine/CoppersmithHH99}, states that the groups in the optimal solution can be separated by hyperplanes in $\R^L$. 
We recall it here as it will be used to derive our second lower bound on $\OPT(V)$ contained in Lemma \ref{6-hyperplane-consequence} below. 


\begin{lemma}[Hyperplanes Lemma~\cite{BPKN:92,journals/datamine/CoppersmithHH99}]  \label{lem:hyperplanes}
	Let $I$ be an impurity measure satisfying properties (P0)-(P2).
	If $(V_i)_{i=1,\dots L}$ is an optimal partition of a set of vectors $V$,  then
	there are vectors  $\vvec^{(1)}, \dots \vvec^{(L)} \in \R^k$ such that 
	$\vvec \in V_i$ if and only if $\vvec \cdot \vvec^{(i)} <  \vvec \cdot \vvec^{(j)}$ for each $j \neq i.$	
\end{lemma}

\begin{lemma} \label{6-hyperplane-consequence}
Let $(V,I,L)$ be an instance of \textsc{PMWIP}. Let $\uvec = \sum_{\vvec \in V} \vvec$. It holds that
	\begin{align*}
		\OPT(V) \ge  \min_{ D \in {\cal D} } \sum_{\dvec' \in D} I(\uvec \circ \dvec' ), 
		\label{eq:optFree}
	\end{align*}
\end{lemma}
\begin{proof}
Let $W$ be the multiset of vectors built as follows:
for each $\vvec=(v_1,\ldots,v_k) \in V$ we add the vectors $v_1 \evec_1, \ldots,v_k \evec_k$ to
$W$. Hence, $W$ has $nk$ vectors, all of them with only one non-zero coordinate.


It is not hard to see that for every partition $V^{(1)}, \dots V^{(L)}$ of $V$ 
there is a corresponding partition $W^{(1)}, \dots, W^{(L)}$
such that $\sum_{i=1}^L I(\sum_{\vvec \in V^{(i)}} \vvec) = \sum_{i=1}^L I(\sum_{\wvec \in W^{(i)}}\wvec)$, hence, 
$$\OPT(V) \geq \OPT(W).$$

Let us now employ Lemma \ref{lem:hyperplanes} to analyze $\OPT(W).$ Let $W^{(1)}, \dots, W^{(L)},$ 
be a partition of $W$ with impurity $\OPT(W)$. From Lemma \ref{lem:hyperplanes}  if two vectors 
$\wvec, \wvec' \in W$ are such that $\wvec = w \evec_i$ and $\wvec' = w' \evec_i$ for some $i$ (i.e., they have the same non-zero component) 
then there is  a $j$ such that both $\wvec$ and $\wvec'$ belong to  $W^{(j)}$.

For $j = 1, \dots, L,$ let $\dvec^{(j)}$ be the vector in $\{0,1\}^k$ such that $\dvec^{(j)}_i = 1$ if and only if
the vectors of $W$ whose only non-zero coordinate is the $i$th one are in $W^{(j)}.$ Then 
$\{\dvec^{(1)}, \dots, \dvec^{(L)}\} \in {\cal D}$ and 
we have $$\OPT(W) = \sum_{i=1}^L I(\sum_{\wvec \in W} \wvec \circ \dvec^{(i)}) =
\sum_{i=1}^L I(\sum_{\vvec \in V} \vvec \circ \dvec^{(i)}) 
= \sum_{i=1}^L I(\uvec \circ \dvec^{(i)}) 
\geq \min_{D \in {\cal D}} \sum_{\dvec' \in D} I(\uvec \circ \dvec').$$  
\end{proof}

Putting together (\ref{eq:collapse}) and Lemmas \ref{lemma:optCollapse}, \ref{6-hyperplane-consequence} we have

\begin{eqnarray} 
\frac{I({\cal A})}{\OPT(V)} &\le& 
\frac{ \sum_{i=1}^L  I \left ( (\uvec^{(i)} \cdot \dvec^{(1)},\ldots,\uvec^{(i)} \cdot \dvec^{(L)}) \right )
+  \sum_{\dvec \in D} I( \uvec \circ \dvec)}{\max \left \{ \OPT(collapse_{D}(V)) ,\min_{ D \in {\cal D}  } \sum_{\dvec' \in D} I(\uvec \circ \dvec' )\right\}} \nonumber \\
&\le& 
\frac{ \sum_{i=1}^L  I \left ( (\uvec^{(i)} \cdot \dvec^{(1)},\ldots,\uvec^{(i)} \cdot \dvec^{(L)}) \right )}{\OPT(collapse_{D}(V)) }
 + 
\frac{   \sum_{\dvec \in D} I( \uvec \circ \dvec)  }{\min_{ D \in {\cal D}  } \sum_{\dvec' \in D} I(\uvec \circ \dvec' ) } \label{eq:BoundGeneralApproach} 
\end{eqnarray}


Since the first ratio in the last expression is the approximation attained by the partition 
$collapse_D(V^{(1)}), \dots, collapse_D(V^{(L)})$ on the instance $(collapse_D(V), I, L)$,  
this inequality says that  we can obtain a good approximation for instance 
$(V,I,L)$ of \textsc{PMWIP} (where the vectors have dimension $k > L$)
by properly choosing: (i)  a set $D$ of $L$ orthogonal directions in $\{0,1\}^k$,  
and---given the choice of $D$---(ii) a good approximation for
the instance $(collapse_{D}(V), I, L)$, where the vectors have dimension $L$.

\section{The dominance algorithm}
For a vector $\vvec$ we say that $i$ is the dominant component for $\vvec$ if $v_i \geq v_j$ for each $j \neq i.$ 
In such a case we also say that 
$\vvec$ is  $i$-dominant.
For a set of vectors $U$ we say that $i$ is the dominant component in $U$ if $i$ is the dominant component 
for $\uvec = \sum_{\vvec \in U} \vvec.$

Given an instance $(V, I)$ let $\uvec = \sum_{\vvec \in V} \vvec$ and let us assume that, up to reordering of the components, 
it holds that $u_i \geq u_{i-1},$ for $i=1, \dots, k-1.$

Let  ${\cal A}^{\sc Dom}$ be the algorithm that proceeds according to
the following cases:

\begin{itemize}

\item[i] $k >L$. 
${\cal A}^{\sc Dom}$ assigns each vector $\vvec=(v_1,\ldots,v_k) \in V$  to group $i$ 
where $i$ is the dominant component of  vector $\vvec'=(v_1,\ldots,v_{L-1}, \sum_{j=L}^k v_j)$

\item[ii] $k \le L$. ${\cal A}^{\sc Dom}$ assigns each vector $\vvec \in V$ 
to group $i$ where $i$ is  the dominant component of  $\vvec$. 
 
\end{itemize}

The only difference between  cases (i) and (ii) is the 
reduction of dimensionality employed in the former to
aggregate the smallest components with respect to $\uvec$.

Let  $D = \{\dvec^{(1)}, \dots, \dvec^{(L)} \} \in {\cal D}$ where $\dvec^{(i)} = \evec_i$ for $i =1, \dots, L-1$ and $\dvec^{(L)} = \onevec - \sum_{\ell=1}^{L-1} \dvec^{(\ell)}$.
We notice that that vector $\vvec'$ in case (i) is exactly 
 $collapse_D(\vvec)$. Thus, if $k > L$,  we can rewrite (\ref{eq:BoundGeneralApproach}) as 

\begin{equation} \label{eq:algoAdom}
 \frac{I({\cal A}^{\sc Dom}(V))}{\OPT(V)} \le
\frac{I({\cal A}^{\sc Dom}(collapse_D(V))}{\OPT(collapse_{D}(V)) }
 + 
\frac{   \sum_{\dvec \in D} I( \uvec \circ \dvec)  }{\min_{ D \in {\cal D}  } \sum_{\dvec' \in D} I(\uvec \circ \dvec' ) } 
\end{equation}

The next lemma  is useful to prove an upper bound on the approximation of 
${\cal A}^{\sc Dom}$, when $k \le L$.


\begin{lemma} \label{generalanalysis}
Let $(V,I)$ be an instance of  \textsc{$(L,k)$-PMWIP} with $I \in {\cal C}$ and $k \leq L$.  
For a subset $S$ of $V$ let $\uvec^S=\sum_{\vvec \in S} \vvec$.
If there exist positive numbers $\alpha, \beta$ such that for each $S \subseteq V$ we have
$$\beta(\|\uvec^S\|_1 - \|\uvec^S\|_{\infty}) \leq I(\uvec^S) \leq \alpha (\| \uvec^S\|_1 - \|\uvec^S\|_{\infty})$$
then the algorithm ${\cal A}^{\sc Dom}$ guarantees $\alpha/\beta$ approximation, 
i.e., 
$$ \frac{I({\cal A}^{\sc Dom}(V))}{\OPT(V)} \le \frac{\alpha}{\beta}.$$ 
\end{lemma}
\begin{proof}
Let $(V^{(1)}, \dots, V^{(L)})$ be the partition of $V$ returned by ${\cal A}^{\sc Dom}$.
Then, by the superadditivity of $I$ 
$$\frac{I({\cal A}^{\sc Dom})}{\OPT(V)} = \frac{ \sum_{i=1}^k I ( \sum_{ \vvec \in V^{(i)}} \vvec) }{ \sum_{i=1}^k \sum_{\vvec \in V^{(i)} } I ( \vvec) }.
$$
Thus, it is enough to prove that for $i=1,\ldots,k$

$$ \frac{ I ( \sum_{ \vvec \in V^{(i)}} \vvec ) } { \sum_{\vvec \in V^{(i)} } I ( \vvec)   } \le \frac{\alpha}{\beta}
$$

Fix $i \in [k]$ and let $\uvec = \sum_{ \vvec \in V^{(i)} } \vvec$.
By hypothesis, we have $$I ( \uvec ) \le \alpha( \| \uvec \|_1 - \| \uvec \|_{\infty}) \quad \mbox{and} \quad 
I( \vvec ) \ge  \beta (\| \vvec \|_1 - \| \vvec \|_{\infty} ) \mbox{ for every } \vvec \in V^{(i)} .$$

Moreover, by construction, for every vector $\vvec \in V^{(i)}$ we have $\| \vvec \|_{\infty}=v_i$, so that
$ \sum_{\vvec \in V^{(i)}} \| \vvec \|_{\infty}  = \| \uvec \|_{\infty} $.  
Putting everything together we have

$$ \frac{ I ( \sum_{ \vvec \in V^{(i)}} \vvec) }{ \sum_{\vvec \in V^{(i)} } I ( \vvec)    } 
\le \frac{ \alpha( \| \uvec \|_1 - \| \uvec \|_{\infty} ) }{ \sum_{\vvec \in V^{(i)} }   \beta (\| \vvec \|_1 - \| \vvec \|_{\infty}) } = 
\frac{ \alpha ( \| \uvec \|_1 - \| \uvec \|_{\infty} ) }{ \beta(\| \uvec \|_1 - \| \uvec \|_{\infty}) } = \frac{\alpha}{\beta},
$$
as desired.
\end{proof}

\subsection{Analysis of ${\cal A}^{\sc Dom}$ for the Gini impurity measure $I_{Gini}$}
In this section we show that algorithm ${\cal A}^{\sc Dom}$ achieves constant $3$-approximation 
when the impurity measure is $I_{Gini}$. 


The following lemma together with Lemma \ref{generalanalysis} will show that ${\cal A}^{\sc Dom}$ guarantees 
$2$-approximation on instances with $k \leq L.$ 

\begin{lemma} \label{boundsGini}
For a vector  $\vvec \in \R_{+}^k$ we
have
$\displaystyle{\| \vvec\|_1 - \| \vvec\|_{\infty}  \le I_{Gini} ( \vvec) \le  2( \| \vvec\|_1 - \| \vvec\|_{\infty})}$.
\end{lemma}
\begin{proof}
First we prove the upper bound.
We have that
\begin{eqnarray}
I_{Gini} ( \vvec) &=& \|\vvec\|_1 \sum_{i=1}^k \frac{v_i}{\| \vvec \|_1} \left( 1 - \frac{v_i}{\| \vvec \|_1}\right) 
 =  \frac{ \sum_{i=1}^k  \sum_{j \ne i} v_i v_j }{\| \vvec\|_1} \\
 &=& \frac{\|\vvec\|_{\infty}(\|\vvec\|_1 - \|\vvec\|_{\infty}) + \sum_{i | v_i \neq \|\vvec\|_{\infty}} v_i \sum_{j \neq i} v_j}{\|\vvec\|_1} \\
 &=& \frac{2 \|\vvec\|_{\infty}(\|\vvec\|_1 - \|\vvec\|_{\infty}) + \sum_{i | v_i \neq \|\vvec\|_{\infty}} v_i \sum_{j \neq \|\vvec\|_{\infty}, j\neq i} v_j}{\|\vvec\|_1} \\
  &\le&  \frac{ 2 \| \vvec\|_{\infty} (  \| \vvec\|_{1} - \| \vvec\|_{\infty}  )  + (  \| \vvec\|_{1} - \| \vvec\|_{\infty})^2      }{\| \vvec\|_1} \\
  &=& \frac{( \| \vvec\|_{1} - \| \vvec\|_{\infty})( \| 2\| \vvec\|_{\infty} + \vvec\|_{1} - \| \vvec\|_{\infty})}{\| \vvec\|_1} \\
& = & \frac{ (  \| \vvec\|_{1} - \| \vvec\|_{\infty}  ) (   \| \vvec\|_{1} + \| \vvec\|_{\infty})  }{
 {\| \vvec\|_1}} \le  2( \| \vvec\|_{1} - \| \vvec\|_{\infty})  
\end{eqnarray}

For the lower bound we observe that $\sum_{i=1}^k v_i^2 \leq \|\vvec \|_{\infty} \|\vvec \|_{1}.$ Therefore, we have
\begin{equation}
I_{Gini} ( \vvec) = \|\vvec\|_1 \sum_{i=1}^k \frac{v_i}{\| \vvec \|_1} \left( 1 - \frac{v_i}{\| \vvec \|_1}\right) 
= \|\vvec \|_1 - \frac{\sum_{i=1}^k v_i^2}{\|\vvec\|_1} \geq \|\vvec \|_1 - \frac{ \|\vvec \|_{\infty} \|\vvec \|_{1}}{\|\vvec\|_1} = 
\|\vvec \|_1 - \|\vvec \|_{\infty}.
\end{equation}

%
%
%
\end{proof}

\begin{theorem} \label{Gini-L-approx}
Algorithm ${\cal A}^{\sc Dom}$ is a 2-approximation algorithm for instances $(V,I)$ with $I=I_{Gini}$ and $k \leq L.$
\end{theorem}
\begin{proof}
Directly from Lemmas \ref{generalanalysis} and \ref{boundsGini}.
\end{proof}

The following lemma will provide an upper bound (in fact an exact estimate) of the second ratio in (\ref{eq:algoAdom}).

\begin{lemma} \label{Gini-directions}
Fix a vector $\uvec \in \R^k$ such that $u_{i} \geq u_{i+1}$ for each $i=1, \dots k-1$ 
and $D = \{\dvec^{(1)}, \dots \dvec^{(L)}\} \in {\cal D}$ with $\dvec^{(i)} = \evec_i$ for $i=1, \dots, L-1$ and
$\dvec^{(L)} = \sum_{j=L}^k \evec_j = \onevec - \sum_{j=1}^{L-1} \dvec^{(j)}.$
It holds that  
$$\sum_{\dvec \in D} I(\uvec \circ \dvec) = \min_{ D' \in {\cal D} } \left \{
\sum_{\dvec' \in D'} I(\uvec \circ \dvec' ) \right \}$$
\end{lemma}
\begin{proof}
Let $D^* \in {\cal D}$ be such that 
\begin{equation} \label{gini-optimaldirections}
\sum_{\dvec^* \in D^*} I(\uvec \circ \dvec^*) = \min_{ D' \in {\cal D} } \left \{
\sum_{\dvec' \in D'} I(\uvec \circ \dvec' ) \right \}
\end{equation}
and $|D^* \cap D|$ is maximum among all 
$D^*$ satisfying (\ref{gini-optimaldirections}). 

Let us assume for the sake of contradiction that $D^* \neq D$. Let $\hat{\dvec} \in D^*$ such that $\hat{\dvec}_k =1$.
We note that $\hat{\dvec} \ne \dvec^{(L)}$ 
for otherwise we would have $D^*=D$.

Let  $\cvec \in D^* \setminus (D \cup \{\hat{\dvec}\})$ such that for all other $\dvec \in  D^* \setminus (D \cup \{\hat{\dvec}\})$ we have 
$\min\{i \mid \cvec_i=1\} < \min\{i \mid \dvec_i=1\}$, i.e., $\cvec$ is the vector in $D^* \setminus (D \cup \{\hat{\dvec}\})$
with the smallest non-zero component. 

Let $\vvec = \cvec + \hat{\dvec}$ and 
 $i^*$ be the minimum integer such that $\vvec_{i^*}=  1$. 
Note that $i^* \leq L-1,$ for otherwise we would have $D^*\notin {\cal D}$.
Let  $F$ be the set of vectors 
from ${\cal D}$ defined by
$$F= (D^* \setminus \{\hat{\dvec}, \cvec\} ) \cup \{\dvec^{(i^*)},\vvec - \dvec^{(i^*)} \}.$$


The following claim directly follows from \cite[Lemma 4.1]{conf/icml/LMP18}. For the sake of self-containment we defer its proof to the appendix. 

\noindent
{\em Claim.} Fix $\uvec \in \R^k$ such that $u_i \geq u_{i+1}$ for each $i=1, \dots, k-1.$ Let ${\bf z}^{(1)}$ and ${\bf z}^{(2)}$ two orthogonal vectors
from $\{0,1\}^k\setminus\{\nullvec\}.$ Let $i^* = \min\{i \mid \max\{z^{(1)}_i, z^{(2)}_i\} = 1\}$ and 
$\vvec^{(1)} = \evec_{i^*} $ and $\vvec^{(2)} = {\bf z}^{(1)}+ {\bf z}^{(2)} - \evec_{i^*}.$ Then 
$$I(\uvec \circ \vvec^{(1)}) + I(\uvec \circ \vvec^{(2)}) \leq I(\uvec \circ {\bf z}^{(1)}) + I(\uvec \circ {\bf z}^{(1)}).$$

\medskip
By the Claim, we have that 
\begin{eqnarray*}
\sum_{\dvec \in F} I(\uvec \circ \dvec) &=& 
I(\uvec \circ \dvec^{(i^*)}) + I(\uvec \circ (\vvec - \dvec^{(i^*)})) + \sum_{\dvec \in F \cap D^*} I(\uvec \circ \dvec)  \\
&\leq&  I(\uvec \circ \hat{\dvec}) + I(\uvec \circ \cvec) + \sum_{\dvec \in F \cap D^*} I(\uvec \circ \dvec)  = 
\sum_{\dvec \in D^*} I(\uvec \circ \dvec),
\end{eqnarray*}
 hence since $D$ satisfies (\ref{gini-optimaldirections}) we have that $F$ also satisfies (\ref{gini-optimaldirections}).

In addition we observe that $|D \cap F| > |D \cap D^*|$ as  by definition it shares with $D$ all that was shared  by $D^*$ and also
$\dvec^{(i^*)}$.
This would be in contradiction with the maximality of the intersection of $D^*$. Therefore, we must have
$D^*  = D$ which concludes the proof. 
\end{proof}

Putting together inequalities (\ref{eq:algoAdom}), Theorem \ref{Gini-L-approx} and Lemma \ref{Gini-directions} we get that

\begin{theorem}
Algorithm ${\cal A}^{\sc Dom}$ is a linear time  3-approximation for the $I_{Gini}$
\end{theorem}

\subsection{Analysis of ${\cal A}^{\sc Dom}$ for the Entropy impurity measure $I_{Ent}$}

The following lemma will be useful for applying Lemma \ref{generalanalysis} 
to the analysis of the performance of ${\cal A}^{\sc Dom}$ with respect to the 
entropy impurity measure $I_{Ent}.$

\begin{lemma} \label{Entropy-basicbounds}
For a vector  $\vvec \in \R_{+}^k$ we
have
$$ ( \|  \vvec\|_1 - \| \vvec\|_{\infty} ) \log \left ( \frac{ \|  \vvec\|_1 }{\min \{ \| \vvec\|_1 - \| \vvec\|_{\infty}, \| \vvec\|_{\infty} \} }
\right ) \le  I_{Ent}(\vvec) \le
2 ( \|  \vvec\|_1 - \| \vvec\|_{\infty} ) \log \left (  \frac{ k \|  \vvec\|_1 }{ \| \vvec\|_1 - \| \vvec\|_{\infty} }
\right )
 $$.
\end{lemma}
\begin{proof}
Let $i^*$ be an index in $[k]$ such that $v_{i^*} = \| \vvec\|_{\infty}.$ We have that 

\begin{eqnarray} 
I_{Ent}(\vvec) &=& \| \vvec\|_{\infty}  \log \frac{ \|  \vvec\|_1}{\| \vvec\|_{\infty} } + \sum_{i \neq i^*}  v_i \log \frac{ \|  \vvec\|_1}{ v_i} \\
&=& \| \vvec\|_{\infty}  \log \frac{ \|  \vvec\|_1}{\| \vvec\|_{\infty} } +
  ( \|  \vvec\|_1- \| \vvec\|_{\infty} ) \log  \|  \vvec\|_1  - \sum_{i \neq i*} v_i \log v_i. \label{x}
\end{eqnarray}

For the upper bound, we observe that 
the expression in (\ref{x})  is maximum when $v_i=( \|  \vvec\|_1 - \|  \vvec\|_{\infty})/(k-1)$ for $i\neq i^*$.
Thus, 
\begin{equation}
  I_{Ent}(\vvec)	\le \| \vvec\|_{\infty} \log \frac{\|  \vvec\|_1}{ \| \vvec\|_{\infty}} 
+ (\|  \vvec\|_1 - \| \vvec\|_{\infty}) \log \frac{\|  \vvec\|_1}{(\|  \vvec\|_1-\| \vvec\|_{\infty})} 
+ (\|  \vvec\|_1-\| \vvec\|_{\infty}) \log (k-1).\
\label{Lemma9:18Jun}
\end{equation}
To show that this satisfies the desired upper bound, we split the analysis into two cases:

\medskip

\noindent
If $  \| \vvec\|_{\infty} \ge  \frac{\| \vvec\|_1}{2}$ we have that

\begin{eqnarray*}
I_{Ent}(\vvec) &\le& \| \vvec\|_{\infty} \log \frac{\| \vvec\|_1}{ \| \vvec\|_{\infty}} + 
(\| \vvec\|_1-\| \vvec\|_{\infty}) \log \frac{\| \vvec\|_1}{\| \vvec\|_1-\| \vvec\|_{\infty}} 
+ (\| \vvec\|_1-\| \vvec\|_{\infty}) \log (k-1) \\
&\le&
2 (\| \vvec\|_1-\| \vvec\|_{\infty}) \log \frac{\| \vvec\|_1}{(\| \vvec\|_1-\| \vvec\|_{\infty})} + 
  (\| \vvec\|_1-\| \vvec\|_{\infty}) \log (k-1)\\
  &\le &  2 (\| \vvec\|_1-\| \vvec\|_{\infty}) \log \frac{k \| \vvec\|_1}{(\| \vvec\|_1-\| \vvec\|_{\infty})},  
\end{eqnarray*}
where the second inequality follows from 
Proposition \ref{prop:basic-inequality} using $p=(\|\vvec\|_1 - \|\vvec\|_{\infty})/\|\vvec\|_1$.

\medskip

\noindent
If $ \| \vvec\|_{\infty} \le  \frac{\| \vvec\|_1}{2}$ we have that

\begin{eqnarray*}  
I_{Ent}(\vvec) & \le & 2 (\| \vvec\|_1-\| \vvec\|_{\infty})\frac{\log e}{ e} 
+ (\| \vvec\|_1 - \| \vvec\|_{\infty}) \log \frac{\| \vvec\|_1}{\| \vvec\|_1 - \| \vvec\|_{\infty}}  +  (\| \vvec\|_1-\| \vvec\|_{\infty}) \log (k-1)    \\
  &  \le & 2(\| \vvec\|_1 - \| \vvec\|_{\infty}) \log \frac{k  \| \vvec\|_1 }{ \| \vvec\|_1 - \| \vvec\|_{\infty}} .
\end{eqnarray*}
where the first inequality follows from  (\ref{Lemma9:18Jun}) and Proposition
\ref{prop:increasing-new}.

\medskip

\noindent
For the lower bound, consider the same two cases:

\noindent
If $ \| \vvec\|_{\infty} >  \frac{\| \vvec\|_1}{2}$,
the expression in (\ref{x}) is minimum when there is a unique index $j \neq i^*$ such that 
$v_j=\| \vvec\|_1 - \| \vvec\|_{\infty} $ and $v_{i} = 0$ for each $i \in [k]\setminus \{j, i^*\}.$  
Thus,
$$ I_{Ent}(\vvec) \ge   \| \vvec\|_{\infty} \log \frac{\| \vvec\|_1}{ \| \vvec\|_{\infty}}  
+   (\| \vvec\|_1- \| \vvec\|_{\infty}) \log \frac{\| \vvec\|_1}{\| \vvec\|_1 - \| \vvec\|_{\infty}}  \ge   
(\| \vvec\|_1 - \| \vvec\|_{\infty}) \log \frac{\| \vvec\|_1}{\min\{ \| \vvec\|_1 - \| \vvec\|_{\infty}, \| \vvec\|_{\infty}\}}	$$

\smallskip

\noindent
If $\| \vvec\|_{\infty} < \frac{\| \vvec\|_1}{2}$,
the expression in (\ref{x})
is minimum when there exists a set of indexes $A \subseteq [k]$ with $|A| =\lceil \|\vvec\|_1 / \| \vvec\|_{\infty} \rceil -1$ such that 
$v_i = \| \vvec\|_{\infty}$ for each $i \in A$ and (possibly) an index $j \not \in A$ such that 
$v_j = \|\vvec\|_1 - |A| \cdot \| \vvec\|_{\infty}.$
Thus,
$$  I_{Ent}(\vvec) \ge  (\|\vvec\|_1 - \| \vvec\|_{\infty}) \log \frac{\|\vvec\|_1}{\| \vvec\|_{\infty}}  \ge  
(\|\vvec\|_1 - \| \vvec\|_{\infty}) \log \frac{\|\vvec\|_1}{\min\{ \|\vvec\|_1 - \| \vvec\|_{\infty}, \| \vvec\|_{\infty}\}} 	$$
\end{proof}

From the bounds in the previous lemma and Lemma \ref{generalanalysis} we obtain our first guarantee on the 
approximation of algorithm ${\cal A}^{\sc Dom}$ for the Entropy Impurity measure on instances with $k \leq L.$

\begin{theorem} \label{theo:Ient-ksmall} 
Let $(V, I, L)$ be an instance of \textsc{PMWIP} with $I = I_{Ent}$ and $k \leq L$.
Let $p= \log k +\log (\sum_{\vvec \in V } \| \vvec \|_1))$. Then,   ${\cal A}^{\sc Dom}$ guarantees 
a $2p$-approximation on instance  $(V, I, L)$.
\end{theorem}
\begin{proof}
Let $S$ be a subset of $V$ and let $\uvec^S= \sum_{\vvec \in S} \vvec$.
Define $\alpha =1$ and $\beta = 2p $.

If $\|\uvec^S\|_1 = \|\uvec^S\|_{\infty}$ then 
$I(\uvec^S)=0$ so that the conditions of Lemma \ref{generalanalysis}
is satisified.
Otherwise, Lemma \ref{Entropy-basicbounds} guarantees
that  $$ \|\uvec^S\|_1 - \|\uvec^S\|_{\infty}  \le I(\uvec^S) \le 2  ( \|\uvec^S\|_1 - \|\uvec^S\|_{\infty} ) \log ( k \| \uvec^S \|_1 )\le  \|\uvec^S\|_1 - \|\uvec^S\|_{\infty}  2 p.$$
Thus, it follows from Lemma \ref{generalanalysis} that we have a $2p-$approximation.
\end{proof}

\begin{remark}
Let $s$ be a large integer.
The instance
$\{(s,0),(2,1),(0,1)\}$ and $L=2$ shows
that the analysis is tight up to constant factors.
In fact, the impurity of ${\cal A}^{\sc Dom}$ is
larger than $ \log s$ while the impurity of 
the partition that leaves $(s,0)$ alone is 
$ 4$.
\label{rem:tight}
\end{remark}

\begin{theorem}  \label{theo:uniform-Ient-ksmall}
Let Uniform-\textsc{PMWIP} (\textsc{U-PMWIP}) be the variant of \textsc{PMWIP} 
where all vectors have the same $\ell_1$ norm.
We have that  ${\cal A}^{\sc Dom}$ is an $O(\log n + \log k )$-approximation algorithm 
for \textsc{U-PMWIP} with $I = I_{Ent}$ and $k \leq L.$ 
\end{theorem}

\begin{proof}
Let $(V, I, L)$ be an instance of \textsc{U-PMWIP} with $I = I_{Ent}$ and vectors of
dimension $k \leq L$.
Let $(V^{(1)}, \dots, V^{(L)})$ be the partition of $V$ returned by ${\cal A}^{\sc Dom}$.
By the superadditivity of $I$ it holds that 

$$\frac{I({\cal A}^{\sc Dom})}{\OPT(V)} = \frac{ \sum_{i=1}^k I ( \sum_{ \vvec \in V^{(i)}} \vvec) }{ 
\sum_{i=1}^k \sum_{ \vvec \in V^{(i)}} I(\vvec)} 
$$

Thus, it is enough to prove that for $i=1,\ldots,k$

$$ \frac{ I ( \sum_{ \vvec \in V^{(i)}} \vvec) }{
\sum_{ \vvec \in V^{(i)}} I(\vvec) } = O(\log n + \log k ) 
$$
Let $s$ be the $\ell_1$ norm of
all vectors in $V$,  let  $\uvec=\sum_{ \vvec \in V^{(i)}} \vvec$ and
let $ c = \| \uvec \|_{1} -\| \uvec \|_{\infty}$. 
By Lemma \ref{Entropy-basicbounds}, we have that 
$$I( \uvec) \le c  \log \frac{k n s}{c}.$$

Moreover, we have
$$ \sum_{ \vvec \in V^{(i)}} I(\vvec) \ge \max \left \{ c,  c \log s - \sum_{ \vvec \in V^{(i)} } 
( \| \vvec \|_1 - \| \vvec \|_{\infty} )  \log (\| \vvec \|_1 - \| \vvec \|_{\infty}) \right \}$$

If $c \ge s/2$ then we have a $O(\log n  + \log k)$ approximation using $c$ as a lower bound.
If $c < s/2$ we get that
$$ \sum_{ \vvec \in V^{(i)}} I(\vvec) \ge c \log s - c \log c = c \log (s/c)$$
and the approximation is $O(\log n + \log k)$ as well
\end{proof}

\begin{remark}
\label{rem:tight2}
Let $s$ be a large integer. The instance with $n-1$ vectors equal to $(s,0)$, one
vector equals to $(s,s/2)$ and  $L=2$ shows that the analysis is tight.
\end{remark}

To obtain an approximation of ${\cal A}^{\sc Dom}$ for $I_{Ent}$ 
for general $L$ and $k$  we need an upper  bound on the second fraction
in Equation (\ref{eq:BoundGeneralApproach}). This is given by the next lemma.

\begin{lemma}  \label{Entropy-directions}
Fix a vector $\uvec \in \R^k$ such that $u_{i} \geq u_{i+1}$ for each $i=1, \dots k-1$ 
and $D = \{\dvec^{(1)}, \dots \dvec^{(L)}\} \in {\cal D}$ with $\dvec^{(i)} = \evec_i$ for $i=1, \dots, L-1$ and
$\dvec^{(L)} = \sum_{j=L}^k \evec_j = \onevec - \sum_{j=1}^{L-1} \dvec^{(j)}.$
It holds that  
$$\sum_{\dvec \in D} I_{Ent}(\uvec \circ \dvec) \le O( \log L) \min_{ D' \in {\cal D} } \left \{
\sum_{\dvec' \in D'} I_{Ent}(\uvec \circ \dvec' ) \right \}$$
\end{lemma}
\begin{proof}
Let $D^*=\{\dvec_*^{(1)}, \dots, \dvec_*^{(L)} \} \in {\cal D}$ be such that 
\begin{equation} \label{opt-condition-entr}
\sum_{\dvec \in D^*} I_{Ent}(\uvec \circ \dvec) =  \min_{ D' \in {\cal D} } \left \{
\sum_{\dvec' \in D'} I_{Ent}(\uvec \circ \dvec' ) \right \}
\end{equation}
and $|D \cap D^*|$ is maximum among all set of vectors in ${\cal D}$ satisfying (\ref{opt-condition-entr}). Assume that 
$D \neq D^*$ for otherwise the claim holds trivially. 

By Lemma \ref{lemma:subsystem} we have that for every 
$\hat{D} = \{\hat{\dvec}^{(1)}, \dots, \hat{\dvec}^{(L)} \} \in {\cal D}$
\begin{eqnarray*} \sum_{i=1}^L I_{Ent}(\uvec \circ \hat{\dvec}^{(i)}) &=& 
I_{Ent}(\uvec) - I_{Ent}(\uvec \cdot \hat{\dvec}^{(1)}, \dots, \uvec \cdot \hat{\dvec}^{(L)})\\
&=& \|\uvec\|_1 \left(H\left(\frac{u_1}{\|\uvec\|_1}, \dots, \frac{u_k}{\|\uvec\|_1}\right) - 
H\left(\frac{\uvec \cdot \hat{\dvec}^{(1)}}{\|\uvec\|_1}, \dots, \frac{ \uvec \cdot \hat{\dvec}^{(L)}}{\|\uvec\|_1}\right) \right)
\end{eqnarray*}
where $H()$ denotes the Entropy function. 
Let us define  
$H(\hat{D})=  H\left(\frac{\uvec \cdot \hat{\dvec}^{(1)}}{\|\uvec\|_1}, \dots, \frac{ \uvec \cdot \hat{\dvec}^{(L)}}{\|\uvec\|_1}\right)$

Then $\hat{D}$ is a set of vectors that 
minimizes $\sum_{i=1}^L I_{Ent}(\uvec \circ \hat{\dvec}^{(i)})$ iff 
it maximizes  $H(\hat{D}).$

We can think of the vectors in $\hat{D}$ as buckets containing components of $\uvec$, and 
we say that $u_j$ is in bucket $i$ if $\hat{d}^{(i)}_j = 1.$ 
From the above formula and the concavity property of the Entropy function we have that the following claim holds.

\medskip

\noindent
{\em Claim 1.}  Assume that there exists a subset $A\subseteq \{j \mid \hat{d}^{(i)}_j = 1\}$ of bucket $i$ and a subset 
$B \subseteq \{j' \mid \hat{d}^{(i')}_{j'} = 1\}$ of bucket $i'$ such that 
\begin{equation}
\left | \left (\hat{\dvec}^{(i)}\cdot \uvec - \sum_{j \in A} u_j  + \sum_{j' \in B} u_{j'} \right ) - \left  (\hat{\dvec}^{(i')}\cdot \uvec - \sum_{j' \in B} u_{j'} + \sum_{j \in A} u_j \right ) \right |  \leq 
|\hat{\dvec}^{(i)}\cdot \uvec  - \hat{\dvec}^{(i')}\cdot \uvec |
\label{eq:Lemma10-17Jun}
\end{equation}
i.e., swapping bucket for elements in $A$ and $B$ does not increase the absolute difference between the sum of elements in 
buckets $i$ and $i'$. Then, 
for the set of vectors $\tilde{D} = \{\tilde{\dvec}^{(1)}, \dots, \tilde{\dvec}^{(L)} \} \in {\cal D}$ defined by 
$$
\tilde{\dvec}^{(\ell)} = 
\begin{cases}
\hat{\dvec}^{(\ell)} & \ell \not \in \{i, i'\} \\
\hat{\dvec}^{(i)} - \sum_{j \in A} \evec_j + \sum_{j' \in B}\evec_{j'} & \ell = i \\
\hat{\dvec}^{(\ell)} - \sum_{j' \in B} \evec_{j'} + \sum_{j \in A} \evec_{j} & \ell = i',
\end{cases}
$$
i.e., for the set of vectors corresponding to the new buckets, it holds that 
$
H\left(\frac{\uvec \cdot \hat{\dvec}^{(1)}}{\|\uvec\|_1}, \dots, \frac{ \uvec \cdot \hat{\dvec}^{(L)}}{\|\uvec\|_1}\right)
\leq 
H\left(\frac{\uvec \cdot \tilde{\dvec}^{(1)}}{\|\uvec\|_1}, \dots, \frac{ \uvec \cdot \tilde{\dvec}^{(L)}}{\|\uvec\|_1}\right),
$ with the equality holding iff inequality (\ref{eq:Lemma10-17Jun}) is tight.

\medskip

Because of Claim 1, we have that $D^*$ satisfying (\ref{opt-condition-entr}) 
is a set of vectors that coincides with buckets that distribute the components of $\uvec$ in {\em the most
balanced way}, i.e., $H(D^*)$ is maximum among all $D \in {\cal D}.$

From these observations, we can characterize the structure of buckets of $D^*$. For the sake of a simpler notation, let us 
denote with $S^{(i)}$ the sum of components in bucket $\dvec_*^{(i)},$ i.e., 
$S^{(i)} = \uvec \cdot \dvec_*^{(i)}.$ We have the following

\medskip 
\noindent
{\em Claim 2.} The set $D^*$ satisfies the following properties:
\begin{itemize}
\item[(i)] there is no bucket $i$ that consists of a single  element $u_j$ with $j \geq L$;
\item[(ii)] if $u_j$ is not alone in bucket $i$ then for each $i' \neq i$ it holds that $S^{(i')} \geq u_j$;
\item[(iii)] if $u_j$ is not alone in bucket $i$ then for each $i' \neq i$ it holds that $S^{(i')} \geq S^{(i)} - u_j$;
\end{itemize}

For (i), assume, by contradiction that such $i$ and $j$ exists. Then, since $D^* \neq D$, 
there exists a bucket $i' \neq i$  that contains at least two elements, with one of them being $u_{j'}$  for some $j' < L$.
Then, by Claim 1, swapping the buckets for $u_j$ and $u_{j'}$ produces a new set of vectors with 
entropy not smaller than $H(D^*)$ and intersection with $D$ larger than that of  $D^*$, which is a contradiction.

For (ii), we observe that if there exists a bucket $i'$ such that $S^{(i')} < u_j$ by moving every element of 
bucket $i'$ into bucket $i$ and moving only $u_j$ from bucket $i$ into bucket $i'$, by Claim 1, we get a
new set of vectors with entropy larger than $H(D^*)$, which is a contradiction.

For (iii), we observe that if there exists a bucket $i'$ such that $S^{(i')} < S^{(i)} - u_j$ swapping all the elements of 
bucket $i'$ with all the elements of bucket $i$ except for $u_j$, by Claim 1, we get a
new set of vectors with entropy larger than $H(D^*)$, which is a contradiction.

We are now ready to prove the statement of the lemma. From the definition of $D$, 
since for $i=1, \dots, L-1$ the bucket $i$ contains only one element, we have $I(\uvec \circ \dvec^{(i)}) = 0$. Let $S = \sum_{j \geq L} u_j,$ and
define $i(j)$ to be the bucket of $D^*$ that contains $u_j$, for each $j = 1, \dots, k.$
We have
\begin{equation} \label{I-ent-UB}
\sum_{i=1}^L I_{Ent} (\uvec \circ \dvec^{(i)}) = I_{Ent} (\uvec \circ \dvec^{(L)}) = \sum_{j \geq L} u_j \log \frac{S}{u_j} =
\sum_{\substack{j \geq L \\ u_j \leq S^{(i(j))}/2}} u_j \log \frac{S}{u_j}  + \sum_{\substack{j \geq L \\ u_j > S^{(i(j))}/2}} u_j \log \frac{S}{u_j}  
\end{equation} 
where in the last expression we split the summands according to whether $u_j \geq S^{(i(j))}/2$ or $u_j < S^{(i(j))}/2.$
We will argue  that 
\begin{eqnarray} \label{lemma10:twocases}
\sum_{\substack{j \geq L \\ u_j \leq S^{(i(j))}/2}} u_j \log \frac{S}{u_j} &\mbox{ is }& O(\log L) \sum_{i=1}^L I_{Ent} (\uvec \circ \dvec_*^{(i)}) \label{twocases:1}\\ 
\sum_{\substack{j \geq L \\ u_j > S^{(i(j))}/2}} u_j \log \frac{S}{u_j} &\mbox{ is }&  O(\log L) \sum_{i=1}^L I_{Ent} (\uvec \circ \dvec_*^{(i)}), \label{twocases:2}
\end{eqnarray}
from which the statement of the lemma follows. 

\medskip

\noindent
{\em Proof of Inequality (\ref{twocases:1}).} 

Since 
\begin{equation} \label{I-ent-LB} 
\sum_{i=1}^L I_{Ent} (\uvec \circ \dvec_*^{(i)})  
\geq   \sum_{\substack{j \geq L \\ u_j \leq S^{(i(j))}/2}}  u_j \log \frac{S^{(i(j))}}{u_j}.
\end{equation}
it is enough to show that for each $j \geq L,$ with $u_j \leq S^{(i(j))/2}$, we have 
\begin{equation} \label{desired-ratio}
\frac{u_j \log \frac{S}{u_j}}{u_j \log \frac{S^{(i(j))}}{u_j}} \le \log (4L).
\end{equation} 

The above inequality can be established by showing that  $S \le 2L \cdot S^{(i(j))}$ and, then, using the bound $\frac{\log a}{\log b} \le \log (2a/b)$, which holds whenever $b \ge 2$ and $a \ge b$.

To see that $S \le 2L \cdot S^{(i(j))}$, let $\ell$ be a bucket in $D^*$ containing some $u_{j'}$ for $j' \geq L.$ By 
Claim 2 (i) we have that bucket $\ell$ contains at least two elements. 
Let $e(\ell)$ be the element in bucket $\ell$ of minimum value. Then,  
by Claim 2 (iii), we have 
\begin{equation} \label{intermediate-1}
S^{(i(j))} \geq S^{(\ell)} - u_{e(\ell)} \geq S^{(\ell)} /2,
\end{equation}
where the last inequality follows from the fact that bucket $\ell$ has at least two elements. 
Let $B = \{\ell \mid $ bucket $\ell$ has at least one element $u_{j'}$ with $j' \geq L \}$. Then, we have 
$L S^{(i(j))} \geq \sum_{\ell \in B} S^{(\ell)}/2 \ge S/2$, that gives $S/S^{(i(j))} \leq 2L,$ as desired.

\medskip
\noindent
{\em Proof of Inequality (\ref{twocases:2}}). 

First we argue that we can assume that there exists at most one $j$, with $j \ge L$,  
with $u_j >  S^{(i(j))}/2$. 
In fact, if there exist $j \neq j'$ such that $u_j > S^{(i(j))}/2$ and $u_{j'} > S^{(i(j'))}/2$  
then $i(j') \ne i(j)$ and no element $u_r$, with $r <L$, is either in bucket $i(j)$ or in $i(j')$ 
Hence, by  the pigeonhole principle, there must exist elements $u_r$ and $u_s$, with $r,s < L$
that are both in some bucket $i' \notin \{i(j),i(j')\}.$ Thus, by Claim 1, swapping buckets for $u_r$ and $u_j$ 
we get a new set of vectors $D'$ whose buckets are at least as balanced as those of $D^*$ ($H(D') \geq  H(D^*)$)
and $| D' \cap D |\ge |D^* \cap D|$. However, 
 in $D'$ there is one less index $j$ with $j \ge L$ and $u_j >  S^{(i(j))}/2$.  Thus, by repeating
 this argument,  we eventually obtain a $D'$ satisfying $H(D^*) = H(D')$ (maximum) and  there is at most one $j$ satisfying $u_j > S^{(i(j))}/2.$

We also have that $u_j = u_L.$ For otherwise, if $u_L > u_j$, by the previous observation we have that 
$S^{(i(L))} \geq 2 u_L$ hence swapping $u_L$ and $u_j$ we obtain a more balanced set of vectors $D'$ with $H(D') > H(D^*),$ against the hypothesis
that $H(D^*)$ is maximum. Therefore, we can assume, w.l.o.g., that $j = L$ and $i(L) = L.$

Finally, for each $\ell, \ell' < L$ we can assume that  $u_{\ell}$ and $u_{\ell'}$ are in different buckets.
 For otherwise, swapping buckets for $u_j$ and $u_{\ell} \geq u_j=u_{L} > S^{(L)} - u_L$ we get a new set $D'$ with $H(D') \geq H(D^*)$
and for all $j \geq L, \, u_j \leq S^{(i(j))}/2.$ Then, the desired result would follow because
we already proved that inequality (\ref{twocases:1}) holds.

\medskip
Because of the previous observation we can assume that in $D^*$, up to renaming the buckets, for each $m \in [L]$ the element 
$u_m$ is in bucket $m.$ Let $X_m = S^{(m)} - u_m.$ Note that $u_L + \sum_{m=1}^L X_m = S.$


Then, we have the following lower bound on the impurity of the buckets of $D^*$:
\begin{eqnarray} 
\sum_{i=1}^L I_{Ent}(\uvec \circ \dvec_*^{(i)}) &\geq& \sum_{m=1}^{L} u_m \log \frac{S^{(m)}}{u_m} \geq 
u_L \left( \sum_{m=1}^{L} \log \frac{S^{(m)}}{u_m} \right)\\
&=& u_L  \log \left( \prod_{m=1}^L \frac{(u_m + X_m)}{u_m} \right) =
u_L  \log \left( \frac{(u_L + X_L) \prod_{m=1}^{L-1} (u_m + X_m)}{u_L \prod_{m=1}^{L-1} u_m} \right). 
\end{eqnarray}

On the other hand, because of the standing assumption $j = L$ we can write as  
upper bound on the only summand in the left hand side of (\ref{twocases:2})
$$u_j \log \frac{S}{u_j} = u_L \log \left(\frac{(u_L+ X_L) + \sum_{m=1}^{L-1} X_m}{u_L} \right).$$

Therefore, to prove the bound in  (\ref{twocases:2}) it is enough to show
$$(u_L + X_L) + \sum_{m=1}^{L-1} X_m  \le
 \left ( \frac{(u_L+X_L)\prod_{m=1}^{L-1} (X_m+u_m)}{\prod_{m=1}^{L-1} u_m} \right ).$$

We can now show that this inequality holds by using Claim 2 (ii), which gives 
$(u_L + X_L) \geq u_s$ for each $s < L$ such that $X_s \neq 0.$ Therefore, we have  
\begin{eqnarray*}
\left((X_L + u_L) + \sum_{s=1}^{L-1} X_s \right) \prod_{m=1}^{L-1} u_m &=& 
(X_L + u_L)\prod_{m=1}^{L-1} u_m + \sum_{s=1}^{L-1} X_s \prod_{m=1}^{L-1} u_m \notag\\
&\leq& (X_L + u_L)\prod_{m=1}^{L-1} u_m + 
\sum_{s=1}^{L-1} \frac{u_L+X_L}{u_s} X_s \prod_{m=1}^{L-1} u_m \notag \\
&=& (X_L + u_L)\prod_{m=1}^{L-1} u_m + 
(u_L + X_L) \sum_{s=1}^{L-1} \left ( X_s \prod_{m \in [L-1]\setminus s} u_m \right) \notag \\
&=& (u_L + X_L) \left( \prod_{m=1}^{L-1} u_m +  \sum_{s=1}^{L-1} X_s \prod_{m \in [L-1]\setminus s} u_m \right)\\
&\leq& (u_L + X_L)   \prod_{m=1}^{L-1} (u_m + X_m),
\end{eqnarray*}
which concludes the proof of (\ref{twocases:2}). 

The proof of the lemma is complete. 
\end{proof}

By (\ref{eq:algoAdom}), combining 
the results in the previous lemma with Theorems \ref{theo:Ient-ksmall}, \ref{theo:uniform-Ient-ksmall} and
the fact that $L \le n$,  
we have the following results 
that apply regardless the relation between $k$ and $L$.

\begin{theorem} 
Let $(V, I_{Ent}, L)$ be an instance of \textsc{PMWIP} and 
let $p=\min\{\log L, \log k \} +\log (\sum_{\vvec \in V } \| \vvec \|_1)$. 
Then,   ${\cal A}^{\sc Dom}$ on  instance $(V, I_{Ent}, L)$ guarantees $2p$-approximation. 
\end{theorem}

\begin{theorem}  \label{theo:uniform-Ient}
Let Uniform-\textsc{PMWIP} (\textsc{U-PMWIP}) be the variant of \textsc{PMWIP} 
where all vectors have the same $\ell_1$ norm.
We have that  ${\cal A}^{\sc Dom}$ is an $O( \log k + \log n)$-approximation algorithm 
for \textsc{U-PMWIP} with $I = I_{Ent}$. 
\end{theorem}

\section{An $O(\log^2 (\min\{k,L\}))$-approximation for \textsc{PMWIP} with $I_{Ent}$}
In this section we present our main result on the entropy measure. 
Under the assumption $k \leq L,$ we will show the existence of an $O(\log^2 k)$-approximation 
polynomial time algorithm.
Note that in the light of  Lemma \ref{Entropy-directions} and 
the approach of Section \ref{sec:general} (see, in particular equation (\ref{eq:BoundGeneralApproach})), 
this  implies an $O(\log^2(\min\{k, L\}))$-approximation algorithm for any $k$ and $L.$

Recall that  a vector $\vvec$ is called $i$-dominant if $i$ is the largest component in $\vvec$, i.e., $v_i = \|\vvec\|_{\infty}.$ 
Accordingly, we say that a set of vectors $B$ (often, in this section, referred to as a bucket)  is $i$-dominant if $i$ is 
the largest component in the bucket, 
i.e., $\| \sum_{\vvec \in B} \vvec \|_{\infty} = \sum_{\vvec \in B} v_i.$
We use $dom ( \vvec)$ and $dom(B)$, respectively, to denote the index of the dominant component of vectors $\vvec$ and   $\sum_{\vvec \in B} \vvec $.

We will say that a bucket $B$ is $i$-pure if each vector in $B$ is $i$-dominant. A bucket which is not $i$-pure for any $i$ will be called 
a {\em mixed bucket}.  
Following the bound on the impurity of a vector $\vvec$ given by Lemma \ref{Entropy-basicbounds}, we define the {\em ratio of a vector $\vvec$} as 
$$ratio(\vvec)= \frac{ \| \vvec \|_1 }{\| \vvec \|_1 - \| \vvec \|_{\infty} }.$$
and, accordingly,  the {\em ratio of bucket $B$} as 
$$ratio(B)= \frac{ \| \sum_{\vvec \in B} \vvec \|_1 }{\|  \sum_{\vvec \in B} \vvec \|_1 - \|  \sum_{\vvec \in B} \vvec \|_{\infty} }.$$

Abusing notation, 
for a set of vectors $B$ we will use $\|B\|_1$ to denote $\| \sum_{\vvec \in B} \vvec \|_1$ and 
$\|B\|_{\infty}$ to denote $\| \sum_{\vvec \in B} \vvec \|_{\infty}.$ Moreover, we use 
$B(j)$ to denote the set of the $j$ vectors in $B$ of minimum ratio.
Since in this section we are only focusing on the entropy impurity measure, we will use $I$ to
denote $I_{Ent}$

We will find it useful to employ the following corollary of Lemma \ref{Entropy-basicbounds}.

\begin{corollary} \label{corollary:Entropy-basicbounds}
For a vector  $\vvec \in \R_{+}^k$ and $i \in [k]$ we
have 
\begin{align}
 ( \|  \vvec\|_1 - \| \vvec\|_{\infty} ) \max \left \{1, \log \left ( \frac{ \|  \vvec\|_1 }{\vvec\|_1 - \| \vvec\|_{\infty} } \right ) \right \} \le I_{Ent}(\vvec) \le 2 ( \|  \vvec\|_1 - v_i ) \log \left (  \frac{ 2k \|  \vvec\|_1 }{\| \vvec\|_1 - v_i } \right )
\end{align}
\label{cor:approx}
\end{corollary}
\begin{proof}
The second inequality follows from Lemma \ref{Entropy-basicbounds}
and Proposition \ref{prop:increasing-new}, using $A=2k \|  \vvec\|_1$.
\end{proof}

\subsection{Our Tools}
In this section we discuss the main tools employed to design our algorithms.

The example of Remark \ref{rem:tight}, apart from establishing the tightness of ${\cal A}^{\sc Dom}$ for $I_{Ent}$, also
shows that we cannot obtain a very good partition by just considering those containing only pure buckets.
However, perhaps surprisingly, the situation is different if we allow at most one mixed bucket.
This is formalized in Theorem \ref{teo:main},  
our first and main tool to obtain good approximate solutions for  
instances of $\textsc{PMWIP}.$ 
This structural theorem will be used by our algorithms to restrict the space where a  partition with low impurity is searched.
Its proof, presented in the next section, is reasonably involved: it consists of starting with
an optimal partition an then showing how to exchange vectors from its
buckets so that a new partition ${\cal P}'$ satisfying the desired properties  is obtained.

\begin{theorem} 
There exists a partition  ${\cal P'}$ with the following properties:
(i) it has at most one mixed bucket; 
(ii) if $\vvec$ is  an $i$-dominant vector in the mixed bucket and $\vvec'$ is an $i$-dominant vector of
a $i$-pure bucket, then $ratio(\vvec) \le ratio(\vvec')$;
(iii)  the impurity of ${\cal P}'$ is at an $O( \log^2 k)$ factor from the minimum possible impurity.
\label{teo:main}
\end{theorem}

\remove{The algorithms presented in the next section build a partition satisfying properties (i) and (ii) of Theorem \ref{teo:main} and 
such that its impurity is at most an $O(\log k)$ factor worse than the impurity of 
partition ${\cal P}'$ satisfying all the properties of Theorem \ref{teo:main}.}

Our second  tool is a transformation $\chi^{2C}$ that maps  vectors in $\R^k$
into vectors in $\R^2$. The nice property of this transformation is that it preserves the entropy
of a set of $i$-pure vectors up to an $O( \log k)$ distortion as formalized by Proposition \ref{prop:may23-10}.
Thus, in the light of Theorem \ref{teo:main}, instead of searching for low-impurity partitions of 
$k$-dimensional vectors with at least $L$-1 pure buckets, we can search for those in a $2$-dimensional space.

The transformation $\chi^{2C}$ is defined as follows 
$$\chi^{2C}(\vvec)= \begin{cases}
(\|\vvec\|_{\infty},\|\vvec\|_1-\|\vvec\|_{\infty}) & \mbox{if }
\|\vvec \|_{\infty} \geq \frac{1}{2} \| \vvec \|_1\\
(\|\vvec\|_1/2,\|\vvec\|_1/2) &  \mbox{if } \|\vvec \|_{\infty} < \frac{1}{2} \| \vvec \|_1.
\end{cases}
$$

Let $I_2(B)$ to denote the {\em 2-impurity of the set $B$}, that is,
the impurity of the set of $2$-dimensional vectors obtained by applying $\chi^{2C}$ to each vector in $B$.
We have that

\begin{prop}
\label{prop:may23-10}
Fix $i \in [k]$ and let $B$ be an $i$-pure bucket. It holds that 
$$(1/2) I_2(B)\le I(B) \le 2 I_2(B) +  4 (\log k) \sum_{\wvec \in B} I(\wvec).
$$
\end{prop}

Finally, our  last tool is the following result from \cite{journals/tit/KurkoskiY14}, 
here stated following our notation, that shows that $\textsc{PMWIP}$ 
 can be optimally solved when $k=2$.  

\begin{theorem}[\cite{journals/tit/KurkoskiY14}]
Let $V$ be  a set of 2-dimensional vectors and let $L$ be an integer larger than 1.
There exists a polynomial time algorithm to build a partition of $V$ into $L$ buckets 
with optimal impurity.

In addition, the partition computed by the algorithm satisfies the following property: if $B$ is a bucket in
the partition and 
if $ \vvec \in V \setminus B $ then either $ratio(\vvec) \ge \max_{\vvec' \in B} \{ ratio(\vvec') \}$ or $ratio(\vvec) \le \min_{\vvec' \in B} \{ ratio(\vvec') \}$. 
\label{teo:opt-2classes}
\end{theorem}

\remove{This theorem says that there is an algorithm producing a partition of minimum impurity and also satisfying the property that the 
vectors are distributed in the buckets according to their ratio, i.e., there is a sequence of ratios
$0 = r_0 \leq  r_1 \leq \cdots \leq \dots r_L$ such that  for each $i=1, \dots L,$ bucket $B_i$ contains only  vectors of ratio between 
$r_{i-1}$ and $r_i.$}

\remove{In addition, the partition computed by the algorithm satisfies the following ordering property: if $B$ is a bucket in
the partition and 
if $ \vvec \in V \setminus B $ then either $ratio(\vvec) \ge \max_{\vvec' \in B} \{ ratio(\vvec') \}$ or $ratio(\vvec) \le \min_{\vvec' \in B} \{ ratio(\vvec') \}$. 
}

Motivated by the previous results  we define 
${\cal A}^{\sc 2C}$ as the algorithm that takes as input a set of vectors $B$ and an integer $b$ and produces 
a partition of $B$ into $b$ buckets by executing the following steps:
(i) every vector $\vvec \in B$ is mapped to
$\chi^{2C}(\vvec)$; (ii) the algorithm given  
by Theorem \ref{teo:opt-2classes} is applied over the transformed
set of vectors to distribute them into $b$ buckets; 
(iii) the partition of $B$ corresponding to the partition produced in step 
(ii) is returned.  

Algorithm ${\cal A}^{\sc 2C}$ is employed as a subroutine of
the algorithms presented in the next section. The 
 following property holds for ${\cal A}^{\sc 2C}$.

\begin{prop} \label{prop:good-projection}
Let $B$ be an $i$-pure set of vectors.  The impurity of  the partition ${\cal P}$ constructed by the algorithm ${\cal A}^{2C}$ 
on input $(B, b)$ is at most an $O(\log k)$ factor from the  minimum possible impurity for a partition of  
set $B$ into $b$ buckets.
\end{prop}
\begin{proof}
Let ${\cal P}^*$ be the partition of $B$ into $b$  buckets
with minimum impurity. We have that 
\begin{eqnarray*} I({\cal P}) &\le& 2 I_2({\cal P}) + 4 \log k \sum_{\wvec \in B} I(\wvec)
\le 2 I_2({\cal P}^*) + 4 \log k \sum_{\wvec \in B} I(\wvec) \\
&\le& 
 4  I({\cal P}^*) + 4 (\log k) \sum_{\wvec \in B} I(\wvec) = O(\log k) I({\cal P}^*),
\end{eqnarray*} 
where the first inequality follows from Proposition \ref{prop:may23-10} (applied to each bucket of ${\cal P}$), 
the second one from the optimality of 
${\cal P}$, the third one by Proposition \ref{prop:may23-10} (applied to each bucket of ${\cal P}^*$), and the last one by observing that 
by superadditivity of $I$ we have $I({\cal P}^*) \geq  \sum_{\wvec \in B} I(\wvec).$
\end{proof}

\subsection{Proof of Theorem \ref{teo:main}} 
The proof proceeds in  steps.
Lemma \ref{lem:part1} shows that there exist a partition
 with at most one mixed bucket  whose impurity is $O(\log k)$ factor from $\OPT(V)$.
Next, we explain how to modify this partition in order to obtain a new partition ${\cal P}$
with at most one mixed bucket, impurity limited by $O(\log k) \OPT(V)$
and such that the vectors in its $i$-pure buckets are ordered according to their ratios.
Finally, we show how to modify  ${\cal P}$ so that we obtain a partition 
${\cal P}'$ that satisfies the properties of Theorem \ref{teo:main}.

\begin{lemma}
There exists a partition   with at most one mixed bucket that satisfies: (i) the impurity of the mixed bucket is
at a $O( \log k)$ factor from the optimal  impurity and (ii) the sum of the impurities of the
pure buckets is at most the  optimal  impurity.
\label{lem:part1}
\end{lemma}
\begin{proof}
Let ${\cal P}^*$ be an optimal partition.
If ${\cal P}^*$ has at most one mixed bucket we are done.
Otherwise, let
$B_1,\ldots,B_j$, with $j \ge 2$, be the mixed buckets in 
${\cal P}^*$. We assume w.l.o.g.\ that 
$B_1$ is the bucket with the smallest ratio among the mixed buckets.

For $i=2,\ldots,j$, let $S_i =\{ \vvec | \vvec \in B_i \mbox{ and } dom(\vvec) \ne dom(B_i) \}$.
Let ${\cal P}$ be a new partition obtained from ${\cal P}^*$  by replacing $B_1$ with $B'_1=B_1 \cup S_2 \ldots \cup S_j$ and
$B_i$ with $B'_i=B_i \setminus S_i$, for $i \ge 2$.
It is clear that $B'_1$ is the unique mixed bucket in  ${\cal P}$.

It follows from subadditivity that  $I(B'_i ) \le I(B_i)$ for $i>1$,
which establishes (ii). Thus, in order to complete the proof  
it is enough to establish an upper bound on $I(B'_1)$.

For $i=2,\ldots,j$, let  $\uvec^{(i)} = \sum_{\vvec \in B_i} \vvec$ and $\wvec^{(i)} = \sum_{\vvec \in S_i} \vvec.$
Moreover, let $s_i = \|\wvec^{(i)}\|_1.$ 
Thus,


$$\|\uvec^{(i)}\|_1 - \|\uvec^{(i)}\|_{\infty} = \|\uvec^{(i)}\|_1 - \sum_{\vvec \in B_i} v_{dom(B_i)} 
 \geq |\wvec^{(i)}\|_1 - \sum_{\vvec \in S_i} v_{dom(B_i)}
 \geq \|\wvec^{(i)}\|_1/2 = \frac{s_i}{2},$$
where the leftmost inequality holds because for each $\vvec \in S_i$ we have $dom(\vvec) \neq dom(B_i)$, so that $\|\vvec\|_1/2 \geq v_{dom(B_i)}.$


Therefore, it follows from Corollary \ref{cor:approx} that
\begin{equation} I(B_i) \ge (\|\uvec^{(i)}\|_1 - \|\uvec^{(i)}\|_{\infty}) \max \{ 1, \log (ratio(B_i)) \} \ge \frac{s_i}{2} \max \{ 1, \log (ratio(B_i)) \}, 
\label{eq:23may-6}
\end{equation}
for each $i>1$

We assume w.l.o.g. that $B_1$ is $1$-dominant. Let $\uvec^{(1)} = \sum_{\vvec \in B_1} \vvec$ and 
let $s_1 = \|\uvec^{(1)}\|_1$ and $c_1 =   \|\uvec^{(1)}\|_1 -  \|\uvec^{(1)}\|_{\infty}.$
Again, from Corollary \ref{cor:approx}, we have  
\begin{equation}
I(B_1) \ge c_1  \max \{ 1, \log (ratio(B_1)) \} 
\label{eq:23may-7}
\end{equation}

For $i=2,\ldots,j$ let $c_i = \|\wvec^{(i)}\|_1 - w^{(i)}_1.$ 
Let $\uvec = \sum_{\vvec \in B'_1} \vvec.$ Then, $\uvec = \uvec^{(1)} + \sum_{i=2}^j \wvec^{(i)},$ hence 
$\|\uvec \|_1 = \sum_{i=1}^j s_i$ and $ \|\uvec \|_1 - u_1 = \sum_{i=1}^j c_i.$

By Corollary  \ref{Entropy-basicbounds} (with $i=1$) we have that
\begin{eqnarray}
I(B'_1) &\leq& 2\left( \sum_{i=1}^j c_i \right ) \cdot  \log  \left ( 2 k \frac{\sum_{i=1}^j s_i }{ \sum_{i=1}^j c_i } \right ) \notag \\
&\le& 2\left ( c_1 + \sum_{i=2}^j s_i \right ) \log  \left ( 2k \frac{s_1 + \sum_{i=2}^j s_i }{c_1 + \sum_{i=2}^j s_i} \right ) \\
&\le& 2\left ( c_1 + \sum_{i=2}^j s_i \right ) \log  \left ( 2k \frac{s_1}{c_1} \right )
=  2 \left ( c_1 + \sum_{i=2}^j s_i \right ) \log  ( 2k \cdot ratio(B_1)),
\end{eqnarray}
where the second inequality follows from Proposition \ref{prop:increasing-new}.


Since $ratio(B_1) \le ratio(B_i)$ for $i>1$ we can conclude, by using
the lower bounds (\ref{eq:23may-6}) and (\ref{eq:23may-7}) that
 $I(B'_1) = O(\log k) \sum_{i=1}^j I(B_i)$.
\end{proof}

\bigskip


Using the mapping $\chi^{2C}$ and Proposition \ref{prop:good-projection}, we can derive the following result.

\begin{lemma}
There exists a partition   with the following properties:
(i) it has at most one mixed bucket; (ii) if $B_i$ is a $i$-pure bucket and $\vvec$ is a $i$-dominant
vector that belongs to an $i$-pure bucket different from $B_i$ then either $ratio(\vvec) \ge \max_{\vvec' \in B_i} \{ ratio(\vvec') \}$ or
$ratio(\vvec) \le \min_{\vvec' \in B_i} \{ ratio(\vvec') \}$ and (iii) its  impurity  is 
at a $O( \log k)$ factor from the minimum possible impurity.
\label{lem:part2}
\end{lemma}
\begin{proof}

Let ${\cal P}$ be a partition that satisfies  Lemma \ref{lem:part1}.
Let $V_i$ be the set of $i$-dominant vectors that are not in the mixed bucket. If $V_i \neq \emptyset$
 let $B^1_i,\ldots,B^{t(i)}_i$ be the $i$-pure buckets where they lie.
We replace these $t(i)$ buckets by the $t(i)$ buckets obtained by running algorithm ${\cal A}^{2C}$ for input 
$(V_i,t(i))$. This replacement is applied for every $i$.
It follows from Proposition \ref{prop:good-projection}
 that the total impurity of the pure buckets in the new partition is at most at a $O(\log k)$ factor
from the total impurity of the pure buckets in ${\cal P}$. 

The property (ii) is assured by the structure of  the partition constructed
by Algorithm ${\cal A}^{2C}$. In order to guarantee that the ties are
broken correctly we present the $i$-dominant vector for algorithm ${\cal A}^{2C}$ in the order of their
ratios.
\end{proof}

\remove{Let ${\cal A}'$ an algorithm similar to ${\cal A}$ with the difference
that it employs a transformation $\pi'$ rather than $\pi$ defined as follows:

$$\pi'(\vvec)=\begin{cases} \pi(\vvec) & \mbox{ if } ratio(\vvec) \ge 0.5 \\
\pi(\vvec)+(\epsilon,-\epsilon) & \mbox{ otherwise }
\end{cases},$$
where $\epsilon$ is chosen so as to guarantee that $ratio(\pi'(\vvec))< ratio(\pi'(\vvec'))$ 
if and only if $ratio(\vvec)< ratio(\vvec')$ 
}

Now, we conclude the proof  of Theorem \ref{teo:main}.
Our starting point  is the partition ${\cal P}$ that satisfies items (i)-(iii) of
Lemma \ref{lem:part2}. We show how to obtain a partition ${\cal P'}$ from ${\cal P}$ that satisfies 
the properties of Theorem \ref{teo:main}.

Let $B_{mix}$ be the mixed bucket in ${\cal P}$.  We assume w.l.o.g that $dom(B_{mix})=1$.
Moreover, let $B_i$ be the $i$-pure bucket that contains the  $i$-dominant vectors with the smallest ratios.
In what follows we assume that the vectors in $B_i$ are sorted by increasing order of their ratios so that
by the $j$th first vector in $B_i$ we mean the one with the $j$th smallest ratio. 

Let
$s_{i,p} = \| B_i \|_1$ ($p$ indicates a pure bucket, $i$ indicates the dominance, and $s$ indicates that we 
are considering the total sum of the components of the vectors).  Let $V_{i, mix}$ be the set of $i$-dominant 
vectors in $B_{mix}$, i.e., $V_{i, mix} = \{ \vvec \in B_{mix} \mid dom(\vvec) = i\}.$ 
 
Let $s_{i,mix} = \| V_{i,mix} \|_1$, i.e., $s_{i,mix}$ denotes the total sum of the components of the 
$i$-dominant vectors from bucket $B_{mix}$.

In order to explain the construction of ${\cal P'}$ we need to define $2k$ set of vectors $X_1,Y_1,\ldots,X_k,Y_k$ that will
be moved  among the buckets of ${\cal P}$ to obtain ${\cal P}'$. Those are defined according to the following cases:

\medskip

{\bf case 1.} $s_{i,p} < s_{i,mix}$.

\medskip

{\bf subcase 1.1} $i>1$ ($i$ is not the dominant component of $B_{mix}$).

Let $r_i$ be the largest ratio among the ratios of  the vectors from $B_i$.
In addition,  let  $Y_i=B_i$ and let  $X_i$ be the set of $i$-dominant vectors  from $B_{mix}$ whose  ratios  are larger than $r_i$.


\medskip

{\bf subcase 1.2} $i=1$.
Let $m$ be such that $\|V_{1,mix}(m-1) \|_1  \leq s_{1,p}$ and $\|V_{1,mix}(m) \|_1  > s_{1,p}.$ 

Moreover, let $r_1$ be the ratio of the $m$-th first
$1$-dominant vector of 
$V_{1,mix}$. 
 Let $X_1  =  V_{1,mix} \setminus V_{1,mix}(m-1)$ (the set containing all the
$1$-dominant vector of  $B_{mix}$ but the first $m-1$ ones) 
 and let $Y_1 = \{\vvec \in B_1 \mid ratio(\vvec) < r_1\}$ (the set  containing  every  vector
in $B_{1}$ with ratio smaller than $r_1$).

\medskip

{\bf case 2.}  $s_{i,p} \geq s_{i,mix}.$ 
\medskip


In this case, let $m$ be such that $\|B_i(m-1)\|_1 < s_{i,mix}$ and $\|B_i(m)\|_1 \ge s_{i,mix}.$
Moreover, let $r_i$ be the ratio of the $m$-th vector of $B_i$.
We define  $Y_i = B_i(m-1)$ (the set containing the   $m-1$ first vectors of $B_i$) 
 and  $X_i = \{ \vvec \in V_{i, mix} \mid ratio(\vvec) > r_i\}$ (the set  containing  every $i$-dominant vector
in $B_{mix}$ with ratio larger than $r_i$). 

\medskip

Let $X = \bigcup_{i=1}^k X_i$ and let $Y = \bigcup_{i=1}^k Y_i$.
The partition ${\cal P'}$ is obtained from  ${\cal P}$ by replacing the bucket $B_{mix}$ with the bucket
  $B'_{mix}= (B_{mix} \cup Y)  \setminus X$ and the bucket $B_i$, for every $i$,
with $B'_i=(B_i \cup X_i) \setminus Y_i$.

\begin{lemma} 
The partition ${\cal P'}$ satisfies item (i) and (ii) from Theorem \ref{teo:main}.
\end{lemma}
\begin{proof}
By construction every $i$-dominant vector in $B'_{mix}$ has ratio at most $r_i$ and
every $i$-dominant vector in $V \setminus B'_{mix}$ has ratio at least $r_i$.
\end{proof}

\begin{lemma} 
\label{lem:last-step-01Jul}
The impurity of the partition ${\cal P'}$ is  at most    $O(\log k)$ times larger than that of ${\cal P}$.
\end{lemma}
\begin{proof}
See the Appendix.
\end{proof}

\subsection{The approximation algorithm}
We first present a pseudo-polynomial time algorithm that provides an $O(\log^2 k)$ approximation and then we
show how to convert it into a polynomial time algorithm with the same approximation.
The key idea is to look among the partitions that satisfy the properties of Theorem \ref{teo:main} for one that (roughly speaking)  
minimizes the impurity of its mixed bucket plus the sum of the 2-impurity of its pure buckets. 

\medskip

\noindent{ \bf A special case: no mixed bucket.}
Theorem \ref{teo:main} establishes the existence of a partition ${\cal P}^*$ 
whose impurity is an $O(\log^2 k)$ approximation of the optimum and 
has at most one  mixed bucket. For a better understanding of the strategy at the basis of our algorithm, let us first discuss
how one can efficiently construct a good partition for the case where the partition ${\cal P}^*$, achieving 
the $O(\log^2 k)$ approximation,  has no mixed buckets. 

In this case, 
we can employ algorithm ${\cal A}^{2C}$ to obtain a partition with minimum 2-impurity among those that only have pure buckets.
By Proposition \ref{prop:may23-10} it follows that the impurity of a partition made only of pure buckets, is upper bounded by its 
$2$-impurity plus  $O(\log k)$ times a lower bound on the optimal impurity. 
Proceeding like in the proof of Proposition \ref{prop:good-projection} 
then we can show that the impurity of the partition of minimum 2-impurity is upper bounded by the same upper bound 
on the impurity of ${\cal P}^*$. 
 
The partition with minimum $2$-impurity made only of pure buckets can be obtained by means of dynamic programming. 

To see this,  for each $j = 1, \dots, k$ 
let 
\begin{equation} \label{eq:defVi-Si}
V_j=\{ \vvec | dom(\vvec)=j\} \quad \mbox{ and } \quad S_j=\{ \vvec | dom(\vvec)=j' \mbox{ for some } j' \le j\}
\end{equation}
Moreover, for each $b = 1, \dots, L$ 
let $Q^*(S_j,b)$ be a partition of the vectors of $S_j$ into $b$ pure buckets such that
its 2-impurity, denoted by $\OPT_2(j,b)$, is minimum. 
It is not hard to see that the following recurrence holds:

\begin{equation}
\OPT_2(j,b) =\begin{cases} I_2({\cal A}^{2C}(V_j,b)) & \mbox{ if } j=1 \\
  \min_{ 1 \le  b'  < b-j } \{ I_2({\cal A}^{2C}(V_j,b')) +  \OPT_2(j-1,b-b')\} & \mbox{ if } j>1
	\end{cases}
\label{eq:19may-1}
\end{equation}
where ${\cal A}^{2C}(V_j,b)$ is the  partition of $V_j$ into $b$
buckets obtained by the Algorithm ${\cal A}^{2C}$ discussed in the previous section. 

Thus, if there exists a partition ${\cal P}^*,$ without mixed buckets, for which $I({\cal P^*}) = O(\log^2 k) 
\OPT(V)$, then the impurity of the  partition ${\cal Q}^*(S_k,L)$ constructed by a DP algorithm based on the equation (\ref{eq:19may-1}) satisfies

\begin{eqnarray*} 
I({\cal Q}^*(S_k,L)) &\leq&  I_2({\cal Q}^*(S_k,L)) + O(\log k) \sum_{\vvec \in V} I(\vvec) \\ 
&\leq& I_2({\cal P}^*) + O(\log k) \OPT(V) \leq 2 I({\cal P}^*) + O(\log k) \OPT(V) \leq O(\log^2 k) OPT(V),
\end{eqnarray*}
where the first inequality  in the first line follows from  Proposition \ref{prop:may23-10}, 
 the first inequality in the second line is due to the minimality of the $2$-impurity of ${\cal Q}^*(S_k,L)$ and 
the superadditivity of $I$ impying that  $\sum_{\vvec \in V} I(\vvec)$ is a lower bound on $\OPT(V).$

\medskip

\noindent
{\bf A pseudopolytime algorithm for the general case.} Now, we turn to the case where there exists at most  one mixed bucket in the
 partition given by Theorem \ref{teo:main}. 
Given an instance $(V, I, L)$ of \textsc{PMWIP}, let $C = \sum_{\vvec \in V} \|\vvec\|_1$ and for each $i=1,\dots, k,$ let $V_i$ and $S_i$  be as in  
(\ref{eq:defVi-Si}). For fixed   
$w, i \in [k], \ell \in [|V|], c \in [C], b \in [L]$ let us denote by 
${\cal Q}^*(w,\ell,S_i,b,c)$  a partition
of $S_i$ into $b$ buckets that satisfies the following properties:


\begin{itemize}
\item[a] it has one bucket, denoted by $B^{{\cal Q}^*}$,  that contains exactly $\ell$ vectors that are $w$-dominant;
\item[b] it contains at most one mixed bucket. This mixed bucket, if it exists, is the bucket $B^{{\cal Q}^*}$. 
\item[c]  For every $i$, if $\vvec$ and $\vvec'$ are, respectively,  $i$-dominant vectors in $B^{{\cal Q}^*}$ and  $V \setminus B^{{\cal Q}^*}$;
then $ratio(\vvec) \le ratio(\vvec')$;
\item[d]  the total sum of all but the $w$-component of vectors in $B^{{\cal Q}^*}$ is equal to $c$, 
i.e., $c=\|B^{{\cal Q}^*} \|_1 - (\sum_{\vvec \in B^{{\cal Q}^*}} v_w ) $;
\item[e] the sum of the 2-impurities of the buckets in ${\cal Q}^*(w,\ell,S_j,b,c) \setminus B^{{\cal Q}^*}$ 
is minimum among the partitions for $S_j$ into $b$ buckets  that satisfy the previous items.
\end{itemize}

The algorithm  builds  partitions 
${\cal Q}^* = {\cal Q}^*(w,\ell,S_k,L,c)$ for all  possible combinations of 
$w,\ell$ and $c$ and, then, returns the one  with minimum impurity.

This approach is motivated by the following: Let 
 ${\cal P}^*$ be a partition that contains one mixed bucket, denoted by $B^*_{mix},$ 
and  satisfies the properties of Theorem \ref{teo:main}. For such a partition, let $w^*=dom(B^*_{mix})$,  
$\ell^*$ be the number of $w^*$-dominant vectors in $B^*_{mix}$
and  $c^*  = \|B^*_{mix}\|_1 - \sum_{\vvec \in B^*_{mix}} v_{w^*}$ (the sum of all but the 
$w^*$ component of the vectors in $B^*_{mix}.)$ 
Then, it is possible to prove that the impurity of a partition ${\cal Q}^* = {\cal Q}^*(w^*,\ell^*,S_k,L,c^*)$ is
at an $O(\log k)$ factor from that of ${\cal P}^*$ (see the proof of Theorem \ref{theo:algopseudo} below).
The key observations are: (i)
the impurity of the bucket $B^{{\cal Q}^*}$ of ${\cal Q}^*$ is
at an $O(\log k)$ factor from  that of $B^*_{mix}$ since $\|B^{{\cal Q}^*}\|_1$  
is at most twice  $\|B^*_{mix}\|_1$ and 
$\|B^{{\cal Q}^*}\|_1 - \sum_{\vvec \in B^{{\cal Q}^*}} v_{w^*} = \|B^*_{mix}\|_1 - \sum_{\vvec \in B^*_{mix}} v_{w^*} = c^*$; 
(ii)  the sum of the 2-impurity of the buckets in ${\cal Q}^*\setminus B^{{\cal Q}^*}$ is
at most the sum of the 2-impurity of the buckets ${\cal P}^* \setminus B^*_{mix}$ so that 
their standard impurities differ by not more than a logarithmic factor.

\medskip

\noindent{\bf Building the partitions ${\cal Q}^*(w,\ell,S_i,b,c)$.}
To simplify our discussion let us assume w.l.o.g. that $w=1$.


Let ${\cal Q}^* = {\cal Q}^*(w,\ell,S_i,b,c)$  be a partition that satisfies properties (a)-(e) above and
let $I_2^{pure}({\cal Q}^* ) = 
I_2({\cal Q}^* \setminus B^{{\cal Q}^*} )$ be the total $2$-impurity of the buckets of  ${\cal Q}$ which are surely  pure. 
Moreover, let $V_i(j)$ be the set of the $j$ vectors of $V_i$ of smallest ratio, and 
let $c_i(j)  =  \|V_i(j)\|_1 -  \sum_{\vvec \in V_i(j)} v_1,$ i.e., the total sum of all components but the first of the vectors in 
$V_i(j).$

For $i=1$ we have
\begin{equation}
I^{pure}_2({\cal Q}^*(1, \ell, S_i,b,c)) = \begin{cases}I_2( {\cal A}^{2C}( V_1 \setminus V_1(\ell) ,b-1)), & \mbox{if } c=c_1(\ell) 
 \\ \infty & \mbox{ otherwise} \end{cases}
\label{eq:may27-1}
\end{equation}


For $i>1$ we have 

\begin{equation} I^{pure}_2({\cal Q}^*(1,\ell,S_i,b,c))= \min_{ 0 \le j \le |V_i| \atop 0 \le b' < b } \{ I_2( {\cal A}^{2C}(V_i \setminus V_i(j) ,b')) 
+ I^{pure}_2({\cal Q}^*(1,\ell,S_{i-1},b-b',c- c_i(j))) \} 
\label{eq:may27-2}
\end{equation}


Algorithm  \ref{alg:greedy} relies on  equations (\ref{eq:may27-1}) and (\ref{eq:may27-2}).  
First, at line \ref{line:preprocess}, it preprocesses the partitions generated by algorithm ${\cal A}^{2C}$ 
that are used by these equations.
Next, it runs over the possible combinations  $(w,\ell)$ and, for
each of them, the procedure ${\cal M}$ is called to search for a partition with  impurity smaller than
those found so far.


For a fixed pair $(w,\ell)$, procedure ${\cal M}$ constructs 
 partitions ${\cal Q}^*(w,\ell,S_i,b,c)$ for all the possible combinations of $i$ and $b$ and all the possible  
corresponding  $c$. Thus,  to simplify we use  ${\cal Q}^*(S_i,b,c)$ to refer to ${\cal Q}^*(w,\ell,S_i,b,c)$.
The first step of procedure ${\cal M}$, where component $w$ is relabeled to $1$ is only meant to keep a direct 
correspondence with the assumption $w=1$ in equations (\ref{eq:may27-1}) and (\ref{eq:may27-2}).
Equation (\ref{eq:may27-1}) is implemented in lines \ref{line:U1-beg}-\ref{line:U1-end} to  build the list $U_1$
 that contains all  the partitions ${\cal Q}^*(1,b,c)$ for which $I^{pure}_2({\cal Q}^*(1,b,c)) \ne \infty$.
The loop of lines \ref{line:GenUi-begin}-\ref{line:GenUi-end} calls procedure {\tt GenerateNewList},
that employs Equation (\ref{eq:may27-2}), to build a list $U_i$, from list $U_{i-1}$, containing all partitions  ${\cal Q}^*(i,b,c)$ with $I^{pure}_2() \ne \infty$. 
We note that at line \ref{line:b*} the special bucket $B'$ of the new partition under construction, 
 is obtained as an extension of  the  bucket $B^{{\cal Q}^*}$ of the partition ${\cal Q}^*(i-1, b,)$ in $U_{i-1}$, which 
  includes the vectors in $V_1(\ell)$.

At the end of the procedure ${\cal M}$ the partition of minimum impurity in $U_k$ is returned. This is the 
partition of minimum impurity among the partition ${\cal Q}^*(w, \ell, S_k, b, c)$ stored in list $U_k$ for some $b$ and $c$.
Hence, for $w = w^*$ and $\ell = \ell^*$, in particular, it is a partition that has impurity not larger than  the partition 
${\cal Q}^*(w^*,\ell^*, S_k, L, c^*)$  which we already observed to be an $O(\log k)$ approximation of the minimum impurity 
partition satisfying Theorem \ref{teo:main}.



Since $c \leq C = \sum_{\vvec \in V} \|\vvec\|_1$ and the lists $U_i$ cannot grow larger than $k L C$  
it is easy to see that  the proposed algorithm runs in polynomial time on $n=|V|$ and $ C=\sum_{\vvec \in V} \|\vvec\|_1$.


The following theorem gives a formal proof of the approximation guarantee for the solution 
returned by Algorithm \ref{alg:greedy}


\begin{algorithm}[ht]
\caption{($V$: set of $k$-dimensional vectors; $L$: integer )}
   \label{alg:greedy}

\begin{algorithmic}[1]
{\small
\State Preprocess ${\cal A}^{2C} (V_j \setminus V_j(j'),b)$ for $j=1,\ldots,k$, $j'=1,\ldots,|V_j|$ and $b=1,\ldots,L$ \label{line:preprocess}
\State ${\cal Q}_{Best} \leftarrow $ arbitrarily chosen  partition of $S_k$ into $L$ buckets
\For{$w=1,\ldots,k$ and $\ell=1,\ldots,|V_{w}|$} 
\If {$I({\cal M}(w,\ell)) < I({\cal Q}_{Best})$ }
\State Update ${\cal Q}_{Best}$ to ${\cal M}(w,\ell)$
\EndIf
\EndFor 
\Procedure{${\cal M}$}{$w$:class,$\ell$:integer}
\State Relabel the components of the vectors so that label of component $w$ becomes 1.
\For{$b'=1,\ldots,L$} \label{line:U1-beg}
\State ${\cal Q}  \leftarrow \{V_1(\ell) \} \cup  {\cal A}^{2C}(V_1 \setminus V_1(\ell),b'-1)$ 
\State Add ${\cal Q}$ to $U_1$.  \label{line:U1-end}
\EndFor
\For{$i=2,\ldots,k$}  \label{line:GenUi-begin}
\State $U_i \leftarrow$  {\tt GenerateNewList}$(U_{i-1})$ \label{line:GenUi-end}
\EndFor
\State Return the partition  with minimum impurity in $U_k$

\EndProcedure

\medskip
\Function{{\tt GenerateNewList}}{$U,i$}
  
\For{every partition ${\cal Q}$ in the list $U$} 
\State Let $(i, b, c)$ be the values s.t. ${\cal Q} = {\cal Q}(i, b, c)$
\If{$b < L$}
\For{$b'=1,\ldots,L-b$}
\For{$j=0,\ldots,|V_i|$}
\State $B' \leftarrow B^{\cal Q} \cup V_i(j)$  \label{line:b*}
\State  ${\cal Q}' \leftarrow \{B'\} \cup  ({\cal Q} \setminus B^{\cal Q} )\cup {\cal A}^{2C}(V_i \setminus V_i(j),b')$.
\State{Add ${\cal Q}'$ to $U$} 
\State  $c' \leftarrow \|B'\|_1 - \sum_{\vvec \in B'} v_1$  \label{line:jun1} 
\If{ $U$ contains another ${\cal Q}''$ with parameters $(i,b+b',c')$} 
\If{$I^{pure}_{2}({\cal Q''}) > I^{pure}_{2}({\cal Q}')$} \label{line:imp'}
 \State remove ${\cal Q''}$ from $U$
\Else 
  \State remove ${\cal Q}'$ from $U$  
\EndIf
\EndIf
\EndFor
\EndFor
\EndIf
\EndFor
\Return $U$
\EndFunction
}
\end{algorithmic}	
\end{algorithm}

 \begin{theorem} \label{theo:algopseudo}
For instances with vectors of dimension $k \le L$, there exists a pseudo-polynomial time  $O( \log^2 k)$-approximation algorithm 
for \textsc{PMWIP}.
\label{teo:approx-pseudo}
\end{theorem}
\begin{proof}
Let ${\cal Q}$ be the partition with smallest impurity between the one returned by 
Algorithm \ref{alg:greedy} and the one returned by the DP based algorithm that implements
Equation (\ref{eq:19may-1}).
In addition, 
let ${\cal P}^*$ be a partition that  satisfies the conditions  of Theorem \ref{teo:main}. In particular, we have 
$I({\cal P}^*) \leq O(\log^2 k) \OPT(V).$ 

To show that   
$I({\cal Q})$ is $O(\log^2 k ) \OPT(V)$ we 
compare  $I({\cal Q})$  with $I({\cal P}^*) $.
We argue according to whether ${\cal P}^*$ has a mixed bucket
or not. 

\medskip

\noindent
{\em Case 1.}  ${\cal P}^*$ has a mixed bucket. We can assume that 
${\cal P}^*$ coincides with the partition ${\cal P}'$ of Lemma \ref{lem:last-step-01Jul}.

Let $B'_{mix}$ be the mixed bucket of ${\cal P}'$ and assume w.l.o.g. that  $w'$ is the dominant component in  
$B'_{mix}$. 
Let  $s' = \|B'_{mix}\|_1$, $c' = s' - \| B'_{mix} \|_{\infty}$ and $c = s' - \sum_{\vvec \in B'_{mix}} v_1$ 
(recall that in the proof of Lemma \ref{lem:last-step-01Jul} component $1$ is 
the dominant component of the bucket $B_{mix}$ from  the partition ${\cal P}$ that is used as a basis to obtain ${\cal P'}$; 
note that it is possible to have $1 \neq w'$).
From Proposition \ref{prop:increasing-new}, since $c' \le c$, and the  proof of Lemma \ref{lem:last-step-01Jul} 
we have that
\begin{equation} \label{newbound:1}
2c' \log \frac{2k s'}{c'} \le  2c \log \frac{2k s'}{c} \le O(\log^2 k) \OPT(V).
\end{equation}
In particular, the second  inequality in (\ref{newbound:1}) is proved in 
Appendix D, {\bf Bounds on the mixed buckets $B'_{mix}$}---note that 
with our present definition of $s'$ and $c$ the middle term of  (\ref{newbound:1}) coincides with the right hand side of 
(\ref{eq:17may-1}).

Let $\ell'$ be the number of $w'$-dominant vectors from ${\cal P}'$ that  lie in $B'_{mix}$. 
We know that the impurity of the output partition ${\cal Q}$  
 is not larger than that of   ${\cal Q}^*(w',\ell',S_k,L,c')$, one of the partitions built by Algorithm \ref{alg:greedy}.
Thus, it is enough to show that the impurity of  ${\cal Q}^*(w',\ell',S_k,L,c')$
is at a $O(\log^2 k)$ factor from the optimum. For this we will show that 
$I({\cal Q}^*)$ is $O( I({\cal P}') + \OPT(V) \log k)$. 
In what follows
we use ${\cal Q}^*$ to refer to ${\cal Q}^*(w',\ell',S_k,L,c')$, and as before, 
$B^{{\cal Q}^*}$ denotes the special bucket in ${\cal Q}^*$.

Let  $$s_1 = \| \sum_{\vvec \in B'_{mix} \atop dom(\vvec) = w'} \vvec \|_1 
\,\,\,\,\mbox { and  } \,\,\,\,  c_1 = s_1 -  \sum_{\vvec \in B'_{mix} \atop dom(\vvec) = w'} v_w.$$ 
By Corollary \ref{cor:approx}, with $i=w'$, we have that 
\begin{equation} 
I(B^{{\cal Q}^*}) \le    
2 c' \log \left (  \frac{  2k \cdot (2(c'-c_1) + s_1) }{c'} \right )
\le  2 c' \log \frac{4ks'}{c'}  \leq  O(\log^2 k) \OPT(V)
\label{eq:july8-1}
\end{equation}
where 
\begin{itemize}
\item for the first inequality, we are also using the fact that $\| B^{{\cal Q}^*} \|_1 \le 2(c'-c_1)+s_1$. 
To see that the last relation holds we note that
$$\|\sum_{\substack{\vvec \in B^{{\cal Q}^*} \\ dom(\vvec) \neq w'}} \vvec \|_1 
- \sum_{\substack{\vvec \in B^{{\cal Q}^*} \\ dom(\vvec) \neq w'}} v_{w'}  = c' - c_1,$$
hence $2(c'-c_1)$ 
is an  upper bound on the total mass of the  vectors in $B^{{\cal Q}^*}$ which are not $w'$-dominant.
Therefore, we have the upper bound $2(c'-c_1)+s_1$ used in the first inequality for $\|B^{{\cal Q}^*}\|_1$.
\item for the second inequality we are using $s' \geq c' - c_1 + s_1.$
\item the last inequality follows from (\ref{newbound:1})
\end{itemize}

\bigskip

\remove{
\begin{eqnarray} I({\cal Q}^*)  &=&  I(B^{{\cal Q}^*}) + \sum_{B' \in ( {\cal Q}^* \setminus B^{{\cal Q}^*} )} I(B') 
 \le 2 c \log \left (  \frac{ k \cdot (2(c-c_1) + s_1) }{c} \right )+   \sum_{B' \in ({\cal Q}^* \setminus B^{{\cal Q}^*} )}   I(B') \nonumber \\ 
&\le& 2 c \log \left ( \frac{ k \cdot (2(c-c_1) + s_1) }{c} \right ) + \log k \cdot \left ( \sum_{ B' \in ( {\cal Q}^* \setminus B^{{\cal Q}^*} )}  I_2(B') \right ),  
\label{eq:may24-1}
\end{eqnarray}
}

\bigskip

We now focus on the buckets of ${\cal Q}^*$ different from $B^{{\cal Q}^*}$---which are surely pure. 
From the  proof of Lemma \ref{lem:last-step-01Jul} ( Appendix D, {\bf Bounds on the $i$-pure buckets} )  we have that
the total impurity of the buckets in ${\cal P}'$ different from $B'_{mix}$ satisfies
\begin{equation} \label{newbound:2}
\sum_{B \in {\cal P}' \setminus B'_{mix}} I(B) = O(\log^2 k) \OPT(V).
\end{equation}
 
In addition, we have
\begin{eqnarray}
\sum_{B \in {\cal Q}^* \setminus B^{{\cal Q}^*}} I(B) &\leq& 
\sum_{B \in {\cal Q}^* \setminus B^{{\cal Q}^*}} \big( 2 I_2(B) + 4  (\log k) \sum_{\wvec \in B} I(\wvec) \big) \label{eqnar-jul8:1} \\
&=& 2\sum_{B \in {\cal Q}^* \setminus B^{{\cal Q}^*}} I_2(B) +    4 (\log k)  \sum_{\wvec \in V} I(\wvec)
 \label{eqnar-jul8:2}\\
&\leq& 2\sum_{B \in {\cal P}' \setminus B'_{mix}} I_2(B) + 
4 (\log k) \sum_{\wvec \in V} I(\wvec)
\label{eqnar-jul8:3}\\ 
&\leq& 4 \sum_{B \in {\cal P}' \setminus B'_{mix}} I(B) +  
4 (\log k) \sum_{\wvec \in V} I(\wvec)
\label{eqnar-jul8:4} \\ 
&\leq& O(\log^2) \OPT(V) +   \OPT(V) \log k, \label{eqnar-jul8:5}
\end{eqnarray}
where the inequality in (\ref{eqnar-jul8:1}) follows from Proposition \ref{prop:may23-10};
 (\ref{eqnar-jul8:3}) follows from (\ref{eqnar-jul8:2}) by the property (e); (\ref{eqnar-jul8:4}) follows from (\ref{eqnar-jul8:3}) by Proposition \ref{prop:may23-10}
and, finally, to obtain (\ref{eqnar-jul8:5}) we use (\ref{newbound:2}) for the left term and 
superadditivity for the right term;

From (\ref{eq:july8-1}) and (\ref{eqnar-jul8:1})-(\ref{eqnar-jul8:4}) we have
$$I({\cal Q}) \leq I({\cal Q}^*) = I(B^{{\cal Q}^*}) + \sum_{B \in {\cal Q}^* \setminus B^{{\cal Q}^*}} I(B) = O(\log^2 k) \OPT(V)$$
and the proof for {\em Case 1} is complete. 

\bigskip

\noindent
{\em Case 2.}  ${\cal P}^*$ does not have a mixed bucket.
In this case, let ${\cal Q}'$ be the partition  built according to the recurrence in (\ref{eq:19may-1}).
It was argued right after this inequality that 
$I({\cal Q}')$ is  $O(\log^2 k)OPT(V)$. Thus, $I({\cal Q})$ is  also $O(\log^2 k)OPT(V)$.
 \end{proof}

%

\bigskip

\noindent
{\bf The polynomial time algorithm.}
Let ${\cal P}^*$ be a partition that satisfies the conditions  of Theorem \ref{teo:main}.
If ${\cal P}^*$ does not have a mixed bucket then the DP based on Equation (\ref{eq:19may-1}) is a polynomial time
algorithm that builds a partition whose impurity is at most $O(\log k)$ times larger than that of
${\cal P}^*$. Thus, we  just need to focus in the case where ${\cal P}^*$ has a mixed bucket.

Let {\sc Algo-Prune} be the variant of Algorithm \ref{alg:greedy} that together with the instance takes 
as input an   extra integer parameter $t$ and uses the following additional conditions regarding the way the lists $U_i's$ are handled: 
(i)  only partitions for which the fifth parameter $c$ is at most $ t$ are added to $U_i$; 
(ii)  after creating the list $U_i$ in line \ref{line:GenUi-end} and before proceeding to list $U_{i+1}$ the following pruning is performed: 
the interval $[0,t]$ is split into $4k$ subintervals of length $t/4k$ and while there exist two partitions ${\cal Q}(w, \ell, S_i, b,c )$ 
and ${\cal Q}'(w, \ell, S_i, b, c' )$ in $U_i$ with  
both $c'$ and $c$ lying in the same subinterval,  the one for which the $I^{pure}_2()$  is  larger is removed. 
This step guarantees that  a polynomial number of partitions are kept in $U_i$.

Let us consider the algorithm ${\cal A}_{poly}$ that executes  {\sc Algo-Prune}  
$e = \lceil \log ( \sum_{\vvec \in V} \|\vvec\|_1) \rceil$ times. In the $j$th execution {\sc Algo-Prune} is called with  $t=2^j$. 
After execution $j$ the partition with the minimum impurity found in $U_k$ is kept as ${\cal Q}^{(j)}.$ 
After all the $e$ executions have been performed, the partition with minimum impurity in $\{{\cal Q}^{(1)}, \dots, {\cal Q}^{(e)} \}$
is returned.

From the above observation that in each call of {\sc Algo-Prune} the number of partitions kept in the lists is polynomial in 
size of the instance and the fact that the number of calls to {\sc Algo-Prune} is also polynomial in the size of the input, we have that
${\cal A}_{poly}$ is a polynomial time algorithm for our problem.

It remains to show that ${\cal A}_{poly}$ is also an $O(\log^3 k)$-approximation algorithm. For this, let us consider again the partitions  ${\cal P}^*$ and  
${\cal Q}^*(1,\ell^*,S_k, L, c)$ defined in the case 2 of the proof of Theorem \ref{teo:approx-pseudo}.
We can show that there is a  partition ${\cal Q}$ among 
those constructed by ${\cal A}_{poly}$ such that  $I^{pure}_2({\cal Q}) \le I^{pure}_2( {\cal Q}^*(1, \ell^*, S_k, L, c))$ and
such that the  special bucket $B^{\cal Q}$ of ${\cal Q}$ has $\ell^*$ vectors that are
$1$-dominant and satisfies $\|B^{\cal Q} \|_1 
- \sum_{\vvec \in B^{\cal Q} } v_1 \leq 2 ( \| B_{mix} \|_1 - \| B_{mix} \|_{\infty} ) =2c$. 

Note that these properties are enough to obtain our claim since, with them,  
proceeding as in  the proof of Theorem \ref{teo:approx-pseudo}
one can show that the impurity of  ${\cal Q}$ is at most an $O(\log^3 k)$ factor larger than the optimal impurity.

For the definition of ${\cal Q}$ we need some additional notation. As in Theorem \ref{teo:approx-pseudo}, let  us 
denote with ${\cal Q}^*$ the partition ${\cal Q}^*(1,\ell^*,S_k,L,c).$ Then $B^{{\cal Q}^*}$ denotes the special bucket of this
partition.


For $i=1,\ldots,k$ let $b_i$ be the number of $i$-pure buckets in  ${\cal Q}^* \setminus B^{{\cal Q}^*} $ and 
let $n_i$ be the number of  $i$-dominant vectors that lie in the bucket $B^{{\cal Q}^*}$. 
Moreover, let $c_i  = \|V_i(n_i)\|_1 - \sum_{\vvec \in V_i(n_i)} v_1$. 
With this, we have that $\sum_{i=1}^k c_i=c=\| B_{mix} \|_1 - \| B_{mix} \|_{\infty}$.

The partition ${\cal Q}$ is defined as the last  partition of
the sequence ${\cal Q}_1,\ldots,{\cal Q}_k$, where
\begin{itemize}
 \item ${\cal Q}_1$ is the partition ${\cal Q}^*(1,\ell^*,S_1,b_1,c_1)$ constructed in the 
 $\lceil \log c \rceil$-th 
 execution of {\sc Algo-Prune}, i.e., with $t=2^{\lceil \log c \rceil } > c.$
 \item  For $i > 1$, let ${\cal Q}'_i$ be the partition obtained by  extending  ${\cal Q}_{i-1}$  with
the $b_i$ buckets from the partition ${\cal A}^{2C}(V_i \setminus V_i(n_i), b_i)$ and
replacing the bucket $B^{{\cal Q}_{i-1}}$, from ${\cal Q}_{i-1}$, with  
$B^{{\cal Q}_{i-1}} \cup V_i(n_i)$. Note that such a partition is added to $U_i$ before the pruning step (ii) is executed.
Then,  ${\cal Q}_i$ is defined as the partition that survives (after the pruning step (ii)) in the subinterval where ${\cal Q}'_i$ lies.
\end{itemize}

Let $c'_i = \| B^{{\cal Q}_i} \|_1 - \sum_{\vvec \in B^{{\cal Q}_i}} v_1$ (the total mass of vectors in 
the special bucket $B^{{\cal Q}_i}$ of ${\cal Q}_i$, minus the mass of such vectors in the component $1$).

We can prove by induction  that 
$$\left | c'_i  - \sum_{j=1}^i c_j \right | \le \frac{ i \cdot t}{4k} .$$
For $i=1$ the result holds since $c_1=c'_1$.
It follows from the induction  that 
$$c'_{i-1} -   \sum_{j=1}^{i-1} c_j \le \frac{ (i-1) \cdot t}{4k} $$
The result for $i$ is established by observing that the pruning step (ii) above, ensures that
$$\left | c'_i - ( c'_{i-1} + c_i) \right | \le \frac{t}{4k}$$

Let ${\cal Q}^*_i$ be the subpartition of  
${\cal Q}^*$ that contains bucket $B^{{\cal Q}^*}$ and all $i'$-pure bucket for each  $i' \le i$.
We can also prove by induction that  $I^{pure}_2({\cal Q}_i) \le I^{pure}_2({\cal Q}^*_i).$
For $i=1$ the result holds since ${\cal Q}_1={\cal Q}^*_1$.
For a general $i$ we have that  
\begin{eqnarray*}
I^{pure}_2({\cal Q}_i) &\le& I^{pure}_2({\cal Q}_{i-1}) + {\cal A}^{2C}(V_i \setminus V(n_i), b_i)\\
&\le& I^{pure}_2({\cal Q}^*_{i-1}) + {\cal A}^{2C}(V_i \setminus V(n_i), b_i) = 
I^{pure}_2({\cal Q}^*_{i}).
\end{eqnarray*}


Thus, by using the same arguments employed in the proof of Theorem 
\ref{teo:approx-pseudo} on can show that the impurity of  ${\cal Q}$ is at an $O( \log^2 k)$ factor from the optimal one.
We can now state the main theorem of the paper.

\begin{theorem}
There  is a polynomial time $O(\log^2 (\min\{k, L\}))$ approximation algorithm for \textsc{PMWIP}.
\end{theorem}
\begin{proof}
By the above argument we have that Algorithm ${\cal A}_{poly}$ is 
a polynomial time $O(\log^2 k)$ approximation algorithm for \textsc{PMWIP} with $k \leq L.$ 
For $k > L$, applying  Lemma \ref{Entropy-directions} and 
the approach of Section \ref{sec:general} (see, in particular equation (\ref{eq:BoundGeneralApproach})), 
we have an $O(\log^2 L)$-approximation algorithm. Putting together the two cases we have the claim.  
\end{proof}

\section{Strong Hardness of \textsc{PMWIP} for the Entropy Impurity measure}
In this section we show that \textsc{PMWIP} is {\em strongly} NP-hard when $I$ is the Entropy measure. 
This rules out an $FPTAS$ for 
the problem under the standard complexity assumptions. 
For this we show a reduction from  \textsc{3-PARTITION}.

\begin{theorem}
The \textsc{PMWIP}  for the Entropy impurity measure is strongly NP-Hard.
\end{theorem}

\begin{proof}
Consider an instance ${\cal I}^{3-Par}$ of \textsc{3-Partition} given by a multiset $U = \{u_1, \ldots, u_k\}$ of $k = 3 L$ integers such that 
for each $i=1, \dots k$ it holds that $T/4 < u_i < T/3 $ where $T = (\sum_{i=1}^k u_1)/L.$
The \textsc{3-Partition} problem consists of deciding whether there exists a partition of $U$ into $L$ parts $A^1, \dots, A^L$
such that the sum of the 
elements in each part is equal to $T.$

From  ${\cal I}^{3-Par}$ we can create in polynomial time an instance $(V,I_{Ent}, L)$ of \textsc{PMWIP} as follows:  
for each number $u_i \in U$ add the scaled canonical vector $\vvec_i = u_i \mathbf{e}_i$ to $V$. 
	
Let $(D^1, D^2, \dots, D^L)$ be a partition of $V$ and let $\uvec= \sum_{\vvec \in V} \vvec =(u_1,\ldots,u_k)$.
Let $\dvec^{(i)} \in \{0,1\}^k$ be defined by 
$d^{(i)}_j=1$ iff $\vvec_j \in D^i$ then the impurity of $(D^1, D^2, \dots, D^L)$ is given by

$$\sum_{i=1}^L I_{Ent}( \uvec \circ \dvec^{(i)} ).$$

By the Subsystem Property---which holds with equality for $I_{Ent}$ (see Lemma \ref{lemma:subsystem}---we have
$$\sum_{i=1}^L I_{Ent}( \uvec \circ \dvec^{(i)} ) = I(\uvec) - I\left((\uvec \cdot \dvec^{(1)}, \uvec \cdot \dvec^{(1)}, \dots, 
\uvec \cdot \dvec^{(L)}) \right),$$
 hence, the right hand side is minimized when $I\left( (\uvec \cdot \dvec^{(1)}, \uvec \cdot \dvec^{(1)}, \dots, 
\uvec \cdot \dvec^{(L)}) \right)$ is maximum, i.e., the vector $(\uvec \cdot \dvec^{(1)}, \uvec \cdot \dvec^{(1)}, \dots, 
\uvec \cdot \dvec^{(L)}) $ is as balanced as possible.
  

By the well known properties of the entropy function, we have that 
deciding in polynomial time whether there is a partition of $V$ such that the resulting impurity is 
at most $I(\uvec) - \|\uvec\| \log L$ is equivalent to decide whether there exists a partition of $V$ into 
sets $D^1, \dots, D^L$ such that 
$\uvec \cdot \dvec^{(1)} = \uvec \cdot \dvec^{(2)} = \cdots = \uvec \cdot \dvec^{(L)} = \|\uvec\|_1 / L,$
which is the same as deciding whether there is a partition $A^1, \dots, A^L$ of $U$ 
such that $\sum_{u \in A^i} u = (\sum_{i=1}^k u_1)/L,$ i.e., solving the instance ${\cal I}^{3-Par}$ of 
\textsc{3-Partition}.
This concludes the reduction. 

Thus,  the strong hardness of \textsc{3-Partition} implies the strong hardness of 
\textsc{PMWIP}  for the Entropy impurity measure.
\end{proof}

\bibliographystyle{abbrv}

\newpage

\appendix
{\small

\section{The proof of Fact \ref{fact:gini-entropy}}

\noindent
{\bf Fact} \ref{fact:gini-entropy}.
{\em The Gini impurity measure defined by $I_{Gini}(\uvec) = \|\uvec\|_1 \sum_{i=1}^k \frac{u_i}{\|\uvec\|_1} (1- \frac{u_i}{\|\uvec\|_1})$ 
and the Entropy impurity measure defined by $I_{Ent} =  \|\uvec\|_1 \sum_{i=1}^k \frac{u_i}{\|\uvec\|_1}  \log (\frac{\|\uvec\|_1}{u_i})$ 
belong to $\calC$.
For $f_{Entr}$, a simple inspection shows that (P3)  holds at equality.} 
\par
 \begin{proof}
	The measure $I_{Gini}$ is obtained using the function $f_{Gini}(x) = x (1-x)$, 
	and $I_{Ent}$ is obtained using the function $f_{Entr}(x) = x \log \frac{1}{x}$. 
	Clearly both functions satisfy property (P1), and it is known they also satisfy (P2)~\cite{journals/datamine/CoppersmithHH99}. 
	So it remains to be shown that they satisfy property (P3).
	
	For $f_{Gini}$, (P3) becomes
	\begin{align*}
		p (1-p) \le p(1-q) + p\left(1-\frac{p}{q}\right) ~~~~\forall q \in [p, 1]
	\end{align*}
	which after canceling the $p$'s out and rearranging, is equivalent to $p \ge q + \frac{p}{q} - 1$ for all $q \in [p,1]$, or $p \ge \max_{q \in [p,1]} (q + \frac{p}{q} - 1)$. But the function in the $\max$ is convex in $q$, and hence its maximum is attained at one of the endpoints $q = p$ and $q = 1$; for these endpoints the inequality holds at equality, which then proves the desired property.
	
\medskip	
	
For the function $f_{Entr}(x) = -x \log x$ we have that for any $0 < x \leq y < 1$ it holds that
$$-\frac{x}{y} y \log(y) - y \frac{x}{y} \log(\frac{x}{y}) = -x \log(y) - x \log(x) + x \log(y) = -x \log(x),$$  
showing that $f_{Entr}(x) = -x \log x$ satisfies (P3) with equality. 
\end{proof}

\section{The proof of the Claim in Lemma \ref{Gini-directions}}

\noindent
{\em Claim.} Fix $\uvec \in \R^k$ such that $u_i \geq u_{i+1}$ for each $i=1, \dots, k-1.$ Let ${\bf z}^{(1)}$ and ${\bf z}^{(2)}$ two orthogonal vectors
from $\{0,1\}^k\setminus\{\nullvec\}.$ Let $i^* = \min\{i \mid \max\{z^{(1)}_i, z^{(2)}_i\} = 1\}$ and 
$\vvec^{(1)} = \evec_{i^*} $ and $\vvec^{(2)} = {\bf z}^{(1)}+ {\bf z}^{(2)} - \evec_{i^*}.$ Then 
$$I(\uvec \circ \vvec^{(1)}) + I(\uvec \circ \vvec^{(2)}) \leq I(\uvec \circ {\bf z}^{(1)}) + I(\uvec \circ {\bf z}^{(1)}).$$

\begin{proof}
For the sake of simplifying the notation, let us assume that $i^* =1$. Since $\vvec^{(1)} + \vvec^{(2)} = {\bf z}^{(1)} + {\bf z}^{(2)},$ and the only significant
components are the non-zero components of ${\bf z}^{(1)} + {\bf z}^{(2)}$, for the analysis, we assume without loss of generality that ${\bf z}^{(2)} = \onevec -
{\bf z}^{(1)}.$ Setting $\dvec = {\bf z}^{(1)},$ we 
have to prove that 
$$ I_{Gini}(\uvec \circ \evec_1 ) + I_{Gini}(\uvec \circ (\onevec - \evec_1))
\le  I_{Gini}(\uvec \circ \dvec ) + I_{Gini}(\uvec \circ (\onevec - \dvec)) ,$$
for every $\dvec \in \{0,1\}^k \setminus \nullvec$. 

It follows from the definition of $I_{Gini}(\cdot)$ that
$$ I_{Gini}(\uvec \circ \dvec ) + I_{Gini}(\uvec \circ (\onevec - \dvec)) =
 (\uvec \cdot \dvec) \left ( \frac{ (\uvec \cdot \dvec)^2 - \sum_{i |d_i=1} (u_i)^2  }{ (\uvec \dvec)^2} \right ) +
(\uvec (\onevec- \dvec)) \left ( \frac{ (\uvec (\onevec- \dvec))^2 - \sum_{i | d_i=0} (u_i)^2   }{ (\uvec (\onevec- \dvec))^2} \right )= $$
$$
\| \uvec \|_1
 - \left ( \frac{ \sum_{i |d_i=1} (u_i)^2  }{ \uvec \cdot \dvec } \right ) -
 \left ( \frac{ \sum_{i |d_i=0} (u_i)^2   }{ \uvec (\onevec- \dvec)} \right ) $$

Define  $g(\dvec)$ as the sum of two last terms of the above expression, that is,

$$g(\dvec) = 
\left ( \frac{ \sum_{i|\dvec_i=1} (u_i)^2  }{ \uvec \cdot \dvec } \right ) +
 \left ( \frac{ \sum_{i|\dvec_i=0} (u_i)^2   }{ \uvec (\onevec- \dvec)} \right ) $$

It is enough to prove that  $g(\evec_1) \ge g(\dvec)$
for an arbitrary $\dvec$. For that,
we   assume w.l.o.g. that $d_1=1$ due to the symmetry of $g(\dvec)$ with respect to $\dvec$.

Let 
$$\alpha= \frac{ \sum_{i>1|d_i=1} (u_i)^2  }{  \sum_{i>1|d_i=1 } u_i }
\,\,\,\,\,\,\, \mbox{ and } \,\,\,\,\,\,\,
\beta= \frac{ \sum_{i|d_i=0 } (u_i)^2  }{ \sum_{i|d_i=0 } u_i }$$
Thus,
$$g(\dvec)= \frac{(u_1)^2 + \alpha (\uvec \cdot \dvec -u_1) }{u_1+(\uvec \cdot \dvec-u_1)} + \beta$$ 
Moreover,  we can write $g(\evec_1)$ as a function of $\dvec$ 
$$g(\evec_1)= u_1 + \frac{\alpha ( \uvec \cdot \dvec-u_1) + \beta \uvec (\onevec- \dvec) }{(\uvec \cdot \dvec-u_1) +\uvec (\onevec- \dvec)} $$
The following inequalities will be useful: $\alpha, \beta \le u_1$ since
$u_1 \ge u_i$ for all $i$, $(\uvec \cdot \dvec -u_1)  \ge \alpha$ and $ \uvec (\onevec- \dvec) \ge \beta$.

We need to prove that
$$ g(\evec_1)= u_1 + \frac{\alpha (\uvec \cdot \dvec -u_1) + \beta (\uvec (\onevec- \dvec))}{(\uvec \cdot \dvec-u_1)+\uvec (\onevec- \dvec)} \ge \frac{(u_1)^2 + \alpha (\uvec \cdot \dvec -u_1)}{u_1+(\uvec \cdot \dvec -u_1)} + \beta= g(\dvec),$$
or equivalently,
$$ \frac{ u_1 ( \uvec \cdot \dvec -u_1 )}{u_1+(\uvec \cdot \dvec -u_1)}
 -  \frac{ \alpha( \uvec \cdot \dvec -u_1 ) }{u_1+(\uvec \cdot \dvec -u_1)}
 \ge
 \frac{\beta ( \uvec \cdot \dvec -u_1) }{\uvec \cdot (\onevec -\dvec) +(\uvec \cdot \dvec -u_1)} 
-  \frac{ \alpha (\uvec \cdot \dvec -u_1) }{\uvec \cdot (\onevec -\dvec) +(\uvec \cdot \dvec -u_1)} 
 $$


Simplifying  the terms we need to prove
$$ ( \beta - \alpha)[(\uvec \cdot \dvec -u_1)+u_1 ] \le (u_1 - \alpha)[(\uvec \cdot \dvec-u_1)+\uvec \cdot (\onevec- \dvec)] $$
which is equivalent to
\begin{equation}
 \beta u_1 - \alpha u_1 \le (u_1 - \beta) (\uvec \cdot \dvec -u_1) + (u_1 - \alpha) \uvec \cdot (\onevec- \dvec) 
\label{eq:aux14Dec}, 
\end{equation}
However, because $\alpha, \beta \le u_1$, $(\uvec \cdot \dvec -u_1)  \ge \alpha$ and $ \uvec \cdot (\onevec- \dvec) \ge \beta$, we have
$$(u_1 - \beta) \alpha + (u_1 - \alpha) \beta \le (u_1 - \beta) (\uvec \cdot \dvec -u_1) + (u_1 - \alpha) \uvec \cdot (\onevec- \dvec). $$
Thus, to establish inequality (\ref{eq:aux14Dec}), it is enough to prove that
$$ \beta u_1 - \alpha u_1 \le (u_1 - \beta) \alpha + (u_1 - \alpha) \beta, $$
or, equivalently,
$$ \alpha \beta \le \alpha u_1.$$
The last inequality  holds because $u_1 \ge \beta$.
\end{proof}

}

\section{The proof of Proposition \ref{prop:may23-10}}

\noindent
{\bf Proposition \ref{prop:may23-10}.}
{\em 
Fix $i \in [k]$ and let $B$ be a set of vector in $\R^k$ such that for each $\vvec \in B,$ it holds that 
$\|\vvec\|_{\infty} = v_i,$ i.e., $B$ is $i$-pure. It holds that 
$$\frac{1}{2} I_2(B)\le I(B) \le 2 I_2(B) +  4 (\log k) \sum_{\wvec \in B} I(\wvec).$$
}
%
\par

\begin{proof}
Let us assume w.l.o.g. that $B$ is 1-pure. Let $\vvec$ be the
vector corresponding to $B$, that is, 
$\vvec=\sum_{\vvec' \in B} \vvec'$.
 Moreover, let
$$ \uvec=\sum_{\vvec' \in B} \chi^{2C}(\vvec')$$,
$$ \uvec^L=\sum_{\vvec' \in B : \| \vvec' \|_{\infty} < \|\vvec\|_1/2 } \chi^{2C}(\vvec') $$
and  
$$\uvec^H=\sum_{\vvec' \in B : \| \vvec \|_{\infty} \geq \|\vvec\|_1/2 } \chi^{2C}(\vvec')$$

Note that $\uvec^L$  corresponds to the set of vectors 
for which the dominant component is  affected by transformation 
$\chi^{2C}$. It shall be clear that  $\| \vvec\|_1=\| \uvec \|_1$
and 
$$\| \vvec\|_{\infty} \le  \|\uvec^L \|_{\infty} + \|\uvec^H \|_{\infty} =
\frac{\|\uvec^L \|_1}{2} + \|\uvec^H \|_{\infty} = \|\uvec\|_{\infty}$$

From  Lemma \ref{Entropy-basicbounds} and Corollary  \ref{corollary:Entropy-basicbounds} we have that
\begin{align}
(\|\vvec\|_1 - \|\vvec\|_{\infty} ) \max \left \{ 1, \log \left (  \frac{\|\vvec\|_1}{\|\vvec\|_1 - \|\vvec\|_{\infty} }  \right )  \right \} \le  I(\vvec) \le  2(\|\vvec\|_1 - \|\vvec\|_{\infty} ) \log \left ( \frac{2k\|\vvec\|_1}{\|\vvec\|_1 - \|\vvec\|_{\infty}}  \right ) 
\label{eq:may23}
\end{align}

Let $\alpha = \|\uvec \|_1 - \|\uvec^H\|_{\infty} - \frac{\|\uvec ^L\|_1}{2}$. Then, we have   

\begin{equation}
I(\uvec) = \alpha  \log \left ( \frac{\|\uvec\|_1}{\alpha} \right ) +
\left (\| \uvec \|_1 - \alpha \right)  \log \left ( \frac{\|\uvec\|_1}{ \| \uvec \|_1 - \alpha} \right )
\label{eq:may23-2}
`\end{equation}

Since 
$\|\uvec^H\|_{\infty} +\frac{\|\uvec^L\|_1}{2} \ge \frac{\|\uvec\|_1}{2}$ then $\alpha \leq \frac{\|\uvec\|_1}{2},$  from 
Proposition \ref{prop:basic-inequality} we have $\frac{\|\uvec\|_1 - \alpha}{\|\uvec\|_1} \log \frac{\|\uvec\|_1}{\|\uvec\|_1 - \alpha} \leq
 \frac{\alpha}{\|\uvec\|_1} \log \frac{\|\uvec\|_1}{\alpha}.$ This, together with (\ref{eq:may23-2}) implies that

\begin{equation}
  \alpha \log  \frac{\|\uvec\|_1}{\alpha}  \le I(\uvec) \le 
 2 \alpha \log  \frac{\|\uvec\|_1}{\alpha}.
\label{eq:may23-2-2}
\end{equation}
Now 
we note that $\|\uvec^H\|_{\infty} + \frac{\|\uvec^L\|_1}{2} > \|\vvec\|_{\infty},$ hence

\begin{equation}
\alpha = \|\uvec\|_1- \|\uvec^H\|_{\infty} - \frac{\|\uvec^L\|_1}{2} \leq 
\|\vvec\|_1 -   \|\vvec\|_{\infty}.
\label{eq:may23-4}
\end{equation}
We first focus on the proof of the left bound $\frac{1}{2} I_2(B)\le I(B)$. We split the analysis into two cases

\medskip

\noindent
{\em Case 1.} $\|\vvec\|_1 - \|\vvec\|_{\infty} \leq \frac{\|\vvec\|_1}{e}.$
Then
\begin{equation} \label{prop3-lb-1}
I(\uvec) \leq 2\alpha \log \frac{\|\vvec\|_1}{\alpha} \leq 2(\|\vvec\|_1 - \|\vvec\|_{\infty}) \log \frac{\|\vvec\|_1}{\|\vvec\|_1 - \|\vvec\|_{\infty}} 
\leq 2I(\vvec)
\end{equation}
where the first inequality is from (\ref{eq:may23-2-2}), the second inequality is from Proposition \ref{prop:increasing-new} and the last inequality is 
from (\ref{eq:may23}) (using the hypothesis at the basis of this case).

\medskip

\noindent
{\em Case 2.} $\|\vvec\|_1 - \|\vvec\|_{\infty} > \frac{\|\vvec\|_1}{e}.$ Then,  
\begin{equation}  \label{prop3-lb-2}
I(\vvec) \geq \left(\|\vvec\|_1 - \|\vvec\|_{\infty} \right) \log e \geq \|\vvec\|_1 =  \|\uvec\|_1 \geq I(\uvec)
\end{equation}
where the first inequality is from (\ref{eq:may23}) and the last inequality is because by definition of $I = I_{Ent}$, 
for a $2$ dimensional vector $\uvec$
we have $I(\uvec) \leq \|\uvec\|_1$.

\medskip

From (\ref{prop3-lb-1}) and (\ref{prop3-lb-1}) it immediately follows that 
$\frac{1}{2} I_2(B) = \frac{1}{2} I(\uvec) \le I(\vvec) =  I(B).$

\bigskip

We now focus on the right inequality and show that $I(\vvec) \leq 2 I(\uvec) + 4 (\|\vvec\|_1 - \|\vvec\|_{\infty}) \log k$, from which also the 
last inequality in the statement of the proposition immediately follows. 

%
%
%
%

First, we observe that 
\begin{eqnarray}
2 \left (\|\uvec\|_1 - \|\uvec^H\|_{\infty} - \frac{\|\uvec^L\|_1}{2} \right ) &=&
\|\uvec^L\|_1 + 2( \|\uvec^H\|_1- \|\uvec^H\|_{\infty} )  \ge  (\|\uvec^L\|_1 + ( \|\uvec^H\|_1- \|\uvec^H\|_{\infty} ) 
\label{prop3-ub-0} \\
&=&
\|\uvec\|_1 - \|\uvec^H\|_{\infty} \ge  \|\vvec\|_1 - \|\vvec\|_{\infty} \label{prop3-ub-1}
\end{eqnarray}

Therefore, usign (\ref{prop3-ub-0})(\ref{prop3-ub-1})  we have 
$2 \left (\|\uvec\|_1 - \|\uvec^H\|_{\infty} - \frac{\|\uvec^L\|_1}{2} \right ) \leq \frac{2k\|\uvec\|_1}{e}$, hence
from the right inequality of (\ref{eq:may23}) we get 
\begin{eqnarray}
I(\vvec) &\le& 2 \left (\|\vvec\|_1 - \|\vvec\|_{\infty} \right) \log \left ( \frac{ 2 k \|\vvec\|_1}{\|\vvec\|_1 - \|\vvec\|_{\infty}} \right) \\
&\le& 4 \left (\|\uvec\|_1 - \|\uvec^H\|_{\infty} - \frac{\|\uvec^L\|_1}{2} \right ) 
\log \frac{2 k \|\vvec\|_1}{2 \left (\|\uvec\|_1 - \|\uvec^H\|_{\infty} - \frac{\|\uvec^L\|_1}{2} \right )} \\
&=& 4 \left (\|\uvec\|_1 - \|\uvec^H\|_{\infty} - \frac{\|\uvec^L\|_1}{2} \right ) 
\left(\log \frac{ \|\vvec\|_1}{ \left (\|\uvec\|_1 - \|\uvec^H\|_{\infty} - \frac{\|\uvec^L\|_1}{2} \right )} 
+ \log k \right) \\
&\leq& 2 I(\uvec) +  4 \left (\|\vvec\|_1 - \|\vvec\|_{\infty} \right)  \log k
\end{eqnarray}
where the last inequality follows from the left hand side of  (\ref{eq:may23-2}) together with the definition of $\alpha$ (for the first term)
and from $\|\uvec^H\|_{\infty} - \frac{\|\uvec^L\|_1}{2} \geq  \| \vvec \|_{\infty}$ and $\|\vvec\|_1 = \|\uvec\|_1$ (for the second term).

Since $B$ is $i$-pure, we have that $\|v\|_1 = \sum_{\wvec \in B} \|\wvec\|_1$ and
$\|v\|_{\infty} = \sum_{\wvec \in B} \|\wvec\|_{\infty}$. Then, we have
$$I(\vvec) \leq 2 I(\uvec) + 4 (\log k) \sum_{\wvec \in B} (\|\wvec\|_1 - \|\wvec\|_{\infty}) \leq 
2 I(\uvec) + 4 (\log k) \sum_{\wvec \in B} I(\wvec)$$
where the last inequality follows by Corollary \ref{corollary:Entropy-basicbounds}.

We have then shown  the right inequalities of the statement. The proof of the Proposition is complete.



\end{proof}

\section{Proof of Lemma \ref{lem:last-step-01Jul}}

\begin{proof}
Recall that  $s_{i,mix}$ denotes the total sum of the components of the 
$i$-dominant vectors from bucket $B_{mix}$ and that we assumed $dom(B_{mix})=1$.
We let $c_{i,mix} = s_{i,mix} - \sum_{\vvec \in V_{i,mix}} v_1$, i.e., the total sum of all but the first component of the 
$i$-dominant vectors in $B_{mix}$.
Moreover,  let
$s_{mix}= \sum_{i=1}^k s_{i,mix}$ and $c_{mix}= \sum_{i=1}^k c_{i,mix}$.
Furthermore, let $c_{i,p} = s_{i,p} - \sum_{\vvec \in B_i} v_i$, i.e., the total sum of the non-$i$ components 
of vectors in $B_i$.

It follows from Corollary \ref{cor:approx} that
\begin{equation} I(B_{mix})   \ge 
 c_{mix} \max \left \{ 1 , \log \left (  \frac{ s_{mix} }{c_{mix} } \right ) \right \} 
\label{eq:12may-1}
\end{equation}


Note that for $i>1$, we have $c_{i,mix} \ge (s_{i,mix}) /2$, for otherwise $i$ would not be the dominant component in $V_{i,mix}$. Thus,
we also have that
\begin{equation} I(B_{mix})  \ge c_{mix} \ge c_{1,mix} + \sum_{i=2}^k s_{i,mix}/2 
\label{eq:16may-1}
\end{equation}
Moreover,  if $c_{mix} < s_{mix} /e$, from (\ref{eq:12may-1}) we have that
\begin{equation} I(B_{mix})  \ge 
 c_{mix} \log \left ( \frac{ s_{mix} }{ c_{mix} } \right ) \ge
\left ( c_{1,mix} + \sum_{i=2}^k s_{i,mix}/2 \right ) \log \left ( \frac{ s_{mix} }{ c_{1,mix} + \sum_{i=2}^k s_{i,mix}/2 } \right ), 
\label{eq:16may-2}
\end{equation}
where the last inequality follows  from  (\ref{eq:16may-1}) and  Proposition \ref{prop:increasing-new}.

From Corollary  \ref{cor:approx} we have that 

\begin{equation} I(B_{i})  \ge c_{i,p} 
 \max \left \{ 1, \log \left (  \frac{ s_{i,p} }{c_{i,p }}  \right ) \right \} 
\label{eq:may17-2}
\end{equation}
for every $i$-pure bucket $B_i$

Now we derive upper bounds on $B'_{mix},B'_1,\ldots,B'_k$ and compare them
with the lower bounds given by the previous equations.

\medskip
\noindent  {\bf Bound on the mixed bucket $B'_{mix}$.}

Let $s_{i,mix}^L = \| V_{i,mix} \cap B'_{mix} \|_1,$ this is the total sum of the components of the 
$i$-dominant vectors in $B_{mix} \cap B'_{mix}.$   

Moreover, let $$c_{i,mix}^L = s_{i,mix}^L  - \sum_{\vvec \in V_{i,mix} \cap B'_{mix}} v_1,$$ i.e.,  
the total sum of all but the first components in the $i$-dominant vectors in $B_{mix} \cap B'_{mix}$.

Recall that $Y_i$ is the set of $i$-dominant vectors moved from $B_i$ to $B_{mix}$ in order to obtain
partition ${\cal P}'$.
Let $s_{i,p}^L = \|Y_i\|_1$ and let  $$t_{i,p}^L = s_{i,p}^L - \sum_{\vvec \in Y_i} v_1,$$ i.e., 
the total sum of all but the first components of vectors in $Y_i$.

In these notations, the superscript $L$ is used to remind the reader that these quantities refer to vectors with 'low' ratio.

From Corollary \ref{cor:approx} (with $i=1$) we have
 that
\begin{align} I (B'_{mix}) \le 2 \left (   \sum_{i=1}^k c_{i,mix}^L+t_{i,p}^L  \right ) 
\log  \left (  
\frac{2k  \left ( \sum_{i=1}^k s_{i,mix}^L+s_{i,p}^L \right ) }{ \sum_{i=1}^k c_{i,mix}^L+t_{i,p}^L  }   \right )  
\label{eq:17may-1}
\end{align}

Moreover, we have
\begin{equation}
\sum_{i=1}^k (c_{i,mix}^L+t_{i,p}^L ) \le c_{1,mix}^L+t_{1,p}^L +  \sum_{i=2}^k (s^L_{i,mix}+s^L_{i,p}) 
\label{eq:1}
\end{equation}

Thus, we have that 
\begin{eqnarray}
I(B'_{mix})  &\le& 2 \left ( c_{1,mix}^L+t_{1,p}^L +  \sum_{i=2}^k (s^L_{i,mix}+s^L_{i,p})    \right ) 
\log  \left (  \frac{2k  \left ( \sum_{i=1}^k s_{i,mix}^L+s_{i,p}^L \right ) }{ c_{1,mix}^L+t_{1,p}^L + 
\sum_{i=2}^k (s^L_{i,mix}+s^L_{i,p})   } \right )  \notag \\
&\le&   2\left ( c_{1,mix}^L+ \sum_{i=2}^k (s^L_{i,mix}+s^L_{i,p})    \right ) 
\log  \left (  \frac{4 k s_{mix} }{ c_{1,mix}^L +   \sum_{i=2}^k (s^L_{i,mix}+s^L_{i,p})  } \right ) \notag \\
& & ~~~~~~ +   2 t_{1,p}^L  \log  \left (  \frac{2k \sum_{i=1}^k (s^L_{i,mix}+s^L_{i,p})   }{ t_{1,p}^L +   \sum_{i=2}^k (s^L_{i,mix}+s^L_{i,p})  } \right )    \notag \\
&\le&  2 \left ( c_{1,mix}^L+ 2 \sum_{i=2}^k s_{i,mix}   \right ) 
\log  \left (  \frac{4 k s_{mix} }{ c_{1,mix}^L +  2 \sum_{i=2}^k s_{i,mix}  } \right )  \label{eq:b'-1} \\
& & ~~~~~~ +  2t_{1,p}^L  \log  \left (  \frac{ 2k \sum_{i=1}^k (s^L_{i,mix}+s^L_{i,p})   }{ t_{1,p}^L +   \sum_{i=2}^k (s^L_{i,mix}+s^L_{i,p})  } \right )
\label{eq:2}, 
\end{eqnarray}
where the first inequality follows from inequality (\ref{eq:17may-1}), Proposition \ref{prop:increasing-new} and inequality (\ref{eq:1}); the second inequality follows from $s_{i,p}^L \le s_{i,mix}$ and the 
third inequality from $s_{i,p}^L \le s_{i,mix}$ together with  Proposition \ref{prop:increasing-new}.

We prove that the expression in (\ref{eq:b'-1})-(\ref{eq:2}) is at most an $O( \log k)$ factor of $I(B_{mix})+I(B_1)$.
First, we consider the  term in  (\ref{eq:b'-1}) that  we denote  by $\alpha$. 

If $c_{mix} \ge s_{mix} /e$ then 
\begin{eqnarray}
\alpha &=& \left ( c_{1,mix}^L+ 2 \sum_{i=2}^k s_{i,mix}   \right ) 
\log  \left (  \frac{2 k s_{mix} }{ c_{1,mix}^L +  2 \sum_{i=2}^k s_{i,mix}  } \right )  \notag \\
&\le& \left ( 2 c_{1,mix} + 2 \sum_{i=2}^k s_{i,mix}   \right ) 
\log  \left (  \frac{2 k s_{mix} }{ 2 c_{1,mix} +  2 \sum_{i=2}^k s_{i,mix}  } \right ) \notag\\
&\le&
\left ( 2 c_{1,mix} + 2 \sum_{i=2}^k s_{i,mix}   \right ) 
\log  \left (  \frac{2 k s_{mix} }{ 2 c_{1,mix} +  2 \sum_{i=2}^k c_{i,mix}  } \right ) \notag\\
&=& \left ( 2 c_{1,mix} + 2 \sum_{i=2}^k s_{i,mix}   \right ) 
\log  \left (  \frac{2 k s_{mix} }{ 2 c_{mix}} \right ) \notag\\
&\le&
 \left ( 2 c_{1,mix}  + 2\sum_{i=2}^k s_{i,mix}  \right )  \log (k e) \leq 4  \left ( c_{1,mix}  + \sum_{i=2}^k \frac{s_{i,mix}}{2}  \right )  \log (k e)
\end{eqnarray}
which is at a $O(\log k)$ factor from the lower bound on $I(B_{mix})$
given by inequality (\ref{eq:16may-1}).

On the other hand, if 
$c_{mix} < s_{mix}/e$, then the first term of (\ref{eq:2})
is at $O( \log k)$ factor from lower bound given by  inequality
(\ref{eq:16may-2}), in fact we have
\begin{eqnarray}
\alpha &=& \left ( c_{1,mix}^L+ 2 \sum_{i=2}^k s_{i,mix}   \right ) 
\log  \left (  \frac{2 k s_{mix} }{ c_{1,mix}^L +  2 \sum_{i=2}^k s_{i,mix}  } \right )  \notag \\
&\le& \left ( c_{1,mix} + 2 \sum_{i=2}^k s_{i,mix}   \right ) 
\log  \left (  \frac{2 k s_{mix} }{ c_{1,mix} +  2 \sum_{i=2}^k s_{i,mix}  } \right )  \notag \\
&\le& 4\left( c_{1,mix} +  \sum_{i=2}^k \frac{s_{i,mix}}{2}   \right ) 
\log  \left (  \frac{2 k s_{mix} }{ c_{1,mix} +   \sum_{i=2}^k \frac{s_{i,mix}}{2}  } \right )  \notag,
\end{eqnarray}
where the first inequality follows from Proposition \ref{prop:increasing-new}.


\medskip

Now, we turn to the second term of (\ref{eq:2}), which we will denote here by $\beta$.
We have that 
\begin{eqnarray*}
\beta &=& t_{1,p}^L  \log  \left (  \frac{2k (s_{1,mix}^L+s_{1,p}^L ) + 2k ( \sum_{i=2}^k s^L_{i,mix}+s^L_{i,p}) }
{ t_{1,p}^L +  \sum_{i=2}^k (s^L_{i,mix}+s^L_{i,p})  } \right ) \\
&\le&  t_{1,p}^L  \log  \left ( \max  \left \{ \frac{2k (s_{1,mix}^L+s_{1,p}^L ) }{t_{1,p}^L},2k \right  \} \right ) \\
&\le&  t_{1,p}^L  \log  \left ( \max  \left \{ \frac{4k \cdot s_{1,p} }{t_{1,p}^L},2k \right  \} \right ) \\
&\le& 
c_{1,p}  \log  \left ( \max  \left \{ \frac{4k \cdot s_{1,p}  }{c_{1,p}},2k \right  \} \right ),
\end{eqnarray*}
where the second inequality holds because  $s^L_{1,mix} \le   s_{1,p}$. 
Moreover, since $t_{1,p}^L \le c_{1,p}$ the last inequality holds due to Proposition \ref{prop:increasing-new}.
It is now easy to see that the quantity in
the righthand side of  the last inequality 
is at a $O(\log k)$ factor from the lower bound on the impurity of $B_1$ given by inequality \ref{eq:may17-2}.

We have completed the proof that $I(B'_{mix}) = O(\log k) (I(B_{mix}) + I(B_1))$ as desired.

\bigskip

\noindent  {\bf Bound on $i$-pure buckets.}
Recall that $X_i$ is the set of vectors moved from $B_{mix}$ to $B_i$ in order to obtain partition
${\cal P}'$.
Let $s^H_{i,mix}$ be the total sum of the components of all vectors in 
 $X_i,$ i.e., $s^H_{i,mix} = \| \sum_{\vvec \in X_i} \vvec \|_1$. Let 
$d^H_{i,mix}$ be the total sum of all but the $i$th components over all vectors in $X_i$, i.e., 
$d^H_{i,mix} = s^H_{i,mix} - \sum_{\vvec \in X_i} v_i.$ 
In addition, let $s^H_{i,p}$ be the total sum of the components of the vectors in the set $B_i \setminus Y_i,$
i.e., $s^H_{i,p} = \|\sum_{\vvec \in B_i \setminus Y_i} \vvec \|_1$. Finally, let  
$c^H_{i,p}$ be the total sum of all but the $i$th component over all the vectors in  $B_i \setminus Y_i$ that is 
$c^H_{i,p} = s^H_{i,p} - \sum_{\vvec \in B_i \setminus Y_i} v_i$ and
let $c^L_{i,p} $ be the total sum of all but the $i$th component over all the vectors in  $Y_i$ that is 
$c^L_{i,p}  = s^H_{i,p} - \| Y_i \|_{\infty}$.

\medskip

\noindent
{\bf Case 1.)} $s_{i,p} \ge s_{i,mix}$.
From Lemma \ref{Entropy-basicbounds} we have  that

\begin{equation} \label{eq:12may-0}
 I(B'_i) \le 2(c^H_{i,p}+ d^H_{i,mix} ) \log \left (   \frac{k \cdot ( s^H_{i,mix}  + s^H_{i,p} )}{c^H_{i,p}+ d^H_{i,mix} } \right )  \le 
2 (c^H_{i,p}+ d^H_{i,mix} ) \log \left (   \frac{2k \cdot s_{i,p}}{c^H_{i,p}+ d^H_{i,mix}} \right ) 
\end{equation}


Let  $\vvec$ be the first vector of bucket $B_i$  that is not moved to $B_{mix}$ and 
let $s=\|\vvec\|_1$. In particular,  $\vvec$ is the vector with the smallest ratio in $B_{i}$ among those that are not moved to $B_{mix}$.
Let $c = s - v_i$.
Recall definition of $r_i$ in the construction of ${\cal P}'$. We have  that
 \begin{equation} \label{eq:pure-case1-1} 
 \frac{s^H_{i,mix} }{d^H_{i,mix} } > r_i \ge  \frac{s^L_{i,p} + s }{c^L_{i,p} + c} .
 \end{equation}
and  
\begin{equation} \label{eq:pure-case1-2}
s^H_{i,mix} \le s_{i,mix} \le s^L_{i,p} + s,
\end{equation}
where the second inequality follows by the definition of $X_i$ and $Y_i$ 
under the standing assumption 
$s_{i,p} \ge s_{i,mix}.$  

Thus, from (\ref{eq:pure-case1-1}) and (\ref{eq:pure-case1-2})  we conclude that  $d^H_{i,mix} \le c^L_{i,p} + c \le c_{i,p},$ hence, 
$c^H_{i,p}+ d^H_{i,mix} \leq 2 c_{i,p}.$

Therefore, from (\ref{eq:12may-0}) and Proposition \ref{prop:increasing-new} we have 
\begin{equation}
\label{eq:proof-main-1}
 I(B'_i) \le 2(c^H_{i,p}+ d^H_{i,mix} ) \log \left (   \frac{2k \cdot s_{i,p}}{c^H_{i,p}+ d^H_{i,mix}} \right )
\le 4c_{i,p} \log \left (   \frac{k \cdot s_{i,p}}{c_{i,p}} \right ).
\end{equation}


\bigskip

\noindent
{\bf Case 2.)} $s_{i,p} < s_{i,mix}$.

\medskip

{\bf subcase 2.1)} $i=1$.
From Lemma \ref{Entropy-basicbounds},  we have
$$I(B'_1) \le  2(c_{1,p}^H + c_{1,mix}^H) 
\log \left ( \frac{k(s_{1,p}^H  + s_{1,mix}^H)}{c_{1,p}^H  + c_{1,mix}^H} \right )
\le   2(c_{1,p}^H + c_{1,mix}^H) 
\log \left ( \frac{2k \cdot s_{mix}}{c_{1,p}^H  + c_{1,mix}^H} \right )$$  

Let $\vvec$ be the first vector of $B_{mix}$  that is moved to $B_1$.
Let $s= \|\vvec\|_1$ and let  $c = s - v_1.$
By construction we have that
$ s^H_{1,p} \le s_{1,p} \le s^L_{1,mix} + s $.

Moreover, 
$$ \frac{s_{1,p}^H}{c_{1,p}^H} \ge r_1 \ge \frac{s^L_{1,mix} + s}{c^L_{1,mix} + c},$$
hence $c_{1,p}^H \le c^L_{1,mix} + c $.

Therefore, $c_{1,p}^H + c_{1,mix}^H \leq 2 c_{1,mix}$. Thus, 
by the above inequality on $I(B'_1)$ and Proposition \ref{prop:increasing-new} we have

\begin{equation}
\label{eq:proof-main-2}
I(B'_1) \le  4 c_{1,mix} 
\log \left ( \frac{2k \cdot s_{mix}}{ c_{1,mix}} \right ) \le 
 4 c_{mix} 
\log \left ( \frac{2k \cdot s_{mix}}{ c_{mix}} \right ).
\end{equation}



\bigskip

{\bf subcase 2.2} $i>1$. 

In this case the bucket $B'_i$ is  exactly the set $X_i$.
Thus,  it follows from the subadditivity of $I$ that
\begin{equation}
\label{eq:proof-main-3}
I(B'_i)=I(X_i) \le I(V_{i,mix})
\end{equation}

\remove{by Corollary \ref{corollary:Entropy-basicbounds} we have 
\begin{equation}
\label{eq:proof-main-3}
 I(B'_i) \le 2d^H_{i,mix} 
\log \left (   \frac{2 k \cdot s^H_{i,mix}}{d^H_{i,mix}} \right ) 
\le  2 s^H_{i,mix} 
\log \left (   \frac{2k \cdot s^H_{i,mix}}{s^H_{i,mix}} \right ) = s^H_{i,mix} \log 2k \le   s_{i,mix} \log 2k 
\end{equation}}

Thus, by aggregating the upper bounds given by
Equations 
(\ref{eq:proof-main-1}), (\ref{eq:proof-main-2}) and (\ref{eq:proof-main-3}),
we get that
$$ \sum_{i=1}^k I(B'_i) \le \sum_{i=1}^k 4c_{i,p} \log \left (   \frac{k \cdot s_{i,p}}{c_{i,p}} \right )   +  4 c_{mix} \log \left ( \frac{2k \cdot s_{mix}}{ c_{mix}} \right ) +  \sum_{i=2}^k
    I(V_{i,mix}) $$
The first term is at most $O(\log k) \sum_{i=1}^k I(B_i)$ due to the lower bound
in (\ref{eq:may17-2}). The second term is $O(\log k) I(B_{mix})$  due to the  lower bounds
in (\ref{eq:16may-1}) and (\ref{eq:16may-2}). Finally, the last term is 
at most $I(B_{mix})$ due to the subadditivity of $I$.

\end{proof}

\end{document}